\documentclass[10pt,a4paper,oneside]{article}
\usepackage[utf8]{inputenc}
\usepackage[english]{babel}
\pdfoutput=1

\usepackage{amsmath,amsfonts,amssymb,amsopn,amscd,amsthm}
\allowdisplaybreaks

\usepackage[yyyymmdd,hhmmss]{datetime}

\usepackage{fancyhdr}
\newcommand{\ClassSingleEntry}{\mathsf{SE}}

\newcommand{\ClassSFS}[1]{\mathsf{SFS}(#1)}
\newcommand{\ClassKMarkov}{\mathsf{HMC}(k)}
\newcommand{\CondMutInf}[3]{I(#1;#2|#3)}
\newcommand{\Entropy}[1]{H(#1)}
\newcommand{\EntropyRate}[1]{\overline{H}(#1)}

\newcommand{\InstantsHiddenTraversal}{\mathcal{I}}
\newcommand{\RealisationHiddenTraversal}{I}
\newcommand{\InstantsNonOverLappingTraversal}[3]{\mathcal{N}_{#1}^{\,#2}(#3)}
\newcommand{\InstantsOccupation}[3]{\mathcal{O}_{#1}^{\,#2}(#3)}
\newcommand{\InstantsTraversal}[3]{\mathcal{T}_{#1}^{\,#2}(#3)}
\newcommand{\Kappa}{\mathcal{K}}

\newcommand{\LimN}{\lim_{n\to\infty}}
\newcommand{\Lumping}{(\TransitionMatrix,\LumpingFunction)}
\newcommand{\LumpingFunction}{g}
\newcommand{\LumpingPreimage}{\LumpingFunction^{-1}}
\newcommand{\MeasureX}{\mu}

\newcommand{\OneOverN}{\frac{1}{n}}
\newcommand{\Sequence}[1]{\mathbf{#1}}
\newcommand{\SpaceW}{\mathcal{W}}
\newcommand{\SpaceX}{\mathcal{X}}
\newcommand{\SpaceY}{\mathcal{Y}}
\newcommand{\SpaceZ}{\mathcal{Z}}

\newcommand{\TransitionMatrix}{P}
\newcommand{\TransitionGraph}{G}
\newcommand{\RenderName}[1]{#1} 
\newcommand{\MakeName}[2]{
	\newcommand{#1}{\RenderName{#2}}
}
\MakeName{\NameCoverThomas}{Cover \& Thomas}
\MakeName{\NameCarlyle}{Carlyle}
\MakeName{\NameChebysheff}{Chebysheff}
\MakeName{\NameGurvitsLedoux}{Gurvits \& Ledoux}
\MakeName{\NameIverson}{Iverson}
\MakeName{\NameMarkov}{Markov}
\MakeName{\NameShannon}{Shannon}
\newtheorem{Thm}{Theorem}
\newtheorem{Lem}[Thm]{Lemma}
\newtheorem{Prop}[Thm]{Proposition}
\newtheorem{Cor}[Thm]{Corollary}

\newtheorem{Def}[Thm]{Definition}

\newtheorem{Exam}[Thm]{Example}

\newcommand{\Then}{\,\Rightarrow\,}

\newcommand{\Iff}{\,\Leftrightarrow\,}
\newcommand{\ForAll}{\forall\,}
\newcommand{\Exists}{\exists\,}
\newcommand{\ExistsUnique}{\exists!\,}

\newcommand{\FirstAlign}{\quad\,\,\,\!}

\newcommand{\NatNum}{\mathbb{N}} 
\newcommand{\NatNumZero}{\NatNum_0} 
\newcommand{\IntNum}{\mathbb{Z}}

\newcommand{\RealNum}{\mathbb{R}}

\newcommand{\Set}[1]{{\{#1\}}}
\newcommand{\BigSet}[1]{{\left\{#1\right\}}}
\newcommand{\IntegerSet}[1]{{[#1]}}
\newcommand{\Cardinality}[1]{{|#1|}}

\newcommand{\Iverson}[1]{{[#1]}}

\DeclareMathOperator{\Exponential}{exp}
\DeclareMathOperator{\ld}{ld}
\DeclareMathOperator{\Support}{supp}

\newcommand{\DistributedAs}{\sim}

\newcommand{\Proba}{\mathbb{P}}
\newcommand{\Expect}{\mathbb{E}}

\usepackage{tikz}

\usetikzlibrary{shapes.geometric}

\tikzstyle{vertex}  = [circle, draw, minimum size = 17pt, inner sep = 0pt]
\tikzstyle{vertexlong}  = [ellipse, draw, inner sep = 2pt]
\tikzstyle{marked} = [very thick, color = red]

\newcommand{\selfedge}[2]{\draw[out=#2-30, in=#2+30, looseness=7,->] (#1) to (#1);}

\newcommand{\drawnode}[3]{\node[vertex] (#1) at (#2) {#3};}

\newcommand{\Abstract}{
\begin{abstract}
A lumping of a Markov chain is a coordinate-wise projection of the chain. We characterise the entropy rate preservation of a lumping of an aperiodic and irreducible Markov chain on a finite state space by the random growth rate of the cardinality of the realisable preimage of a finite-length trajectory of the lumped chain and by the information needed to reconstruct original trajectories from their lumped images. Both are purely combinatorial criteria, depending only on the transition graph of the Markov chain and the lumping function. A lumping is strongly k-lumpable, iff the lumped process is a k-th order Markov chain for each starting distribution of the original Markov chain. We characterise strong k-lumpability via tightness of stationary entropic bounds. In the sparse setting, we give sufficient conditions on the lumping to both preserve the entropy rate and be strongly k-lumpable.
\end{abstract}
}

\newcommand{\TitleFull}{Lumpings of Markov chains, entropy rate preservation, and higher-order lumpability}

\newcommand{\TitleShort}{Lumped\ Markov\ chains\ and\ entropy}
\newcommand{\TitlePDF}{Lumped\ Markov\ chains\ and\ entropy}

\newcommand{\AuthorsFull}{Bernhard C. Geiger\footnote{Email: geiger@ieee.org}
                       \& Christoph Temmel\footnote{Email: math@temmel.me}}

\newcommand{\AuthorsShort}{Geiger \& Temmel}

\newcommand{\Keywords}{Keywords:
lumping,
entropy rate loss,
functional hidden Markov model,
strong lumpability,
higher order Markov chain}

\newcommand{\MSC}{MSC 2010: 60J10 (60G17 94A17 60G10 65C40)}

\newcommand{\Head}{
 \maketitle
 \Abstract{}
 \Keywords{}\\
 
 \noindent\MSC{}\\
 
 \noindent This is an extended version of~\cite{Geiger_Temmel__LumpingsOfMarkovChainsEntropyRatePreservationAndHigherOrderLumpability__JAP_2014}.
 \tableofcontents
 \listoffigures
}

\usepackage[pdftex%
 ,bookmarks=true%
 ,bookmarksopen=true%
 ,bookmarksopenlevel=3%
 ,colorlinks=false%
 ,pdftitle=Arxiv\ \pdfdate:\ \TitlePDF{}%
]{hyperref}

\title{\TitleFull{}}
\author{\AuthorsFull{}}
\date{}
\begin{document}

\Head{}

\section{Introduction}
\label{sec:introduction}

The \emph{entropy rate} of a stationary stochastic process is the average number of bits per time step needed to encode the process. A \emph{lumping} of a (stationary) \NameMarkov{} chain is a coordinate-wise projection of the chain by a \emph{lumping function}. The resulting (stationary) \emph{lumped stochastic process} is also called a \emph{functional hidden \NameMarkov{} model}~\cite{Ephraim_Merhav__HiddenMarkovProcesses__TIT_2002}. One can transform every hidden \NameMarkov{} model on finite state and observation spaces into this setting~\cite[Section~IV.E]{Ephraim_Merhav__HiddenMarkovProcesses__TIT_2002}. In general, the lumped process loses the \NameMarkov{} property~\cite{Gurvits_Ledoux__MarkovPropertyForAFunctionOfAMarkovChain_ALinearAlgebraApproach} and has a lower entropy rate than the original \NameMarkov{} chain, due to the aggregation of states~\cite{Pinsker__InformationAndInformationStabilityOfRandomVariablesAndProcesses__Holden_1964,Watanabe_Abraham__LossAndRecoveryOfInformationByCoarseObservationOfStochasticChain__IC_19}.\\

Our first result characterises the structure of entropy rate preserving lumpings of stationary \NameMarkov{} chains over a finite state space. The \emph{realisable preimage} is the set of finite paths in the transition graph associated with the \NameMarkov{} chain having the same image. The key property is the behaviour of the growth of this random set. It is also described by the ability of two such paths, once split, to join again. We document a strong dichotomy between the preservation and loss case: a uniform finite bound on the lost entropy and almost-surely finite growth in the former, and a linearly growing entropy loss and an almost-surely exponential growth in the latter.\\

In particular, a positive transition matrix always implies an entropy rate loss for a non-identity lumping. We state a sufficient condition on a lumping of a \NameMarkov{} chain with non-positive transition matrix to preserve the entropy rate. \NameCarlyle{}'s representation~\cite{Carlyle__IdentificationOfStateCalculableFunctionsOfFiniteMarkovChains__AMS_1967} of a finite-state stationary stochastic process as a lumping of a \NameMarkov{} chain on an at most countable state space fulfils this condition.\\

Lumpings resulting in higher-order \NameMarkov{} chains are highly desirable from a simulation point of view. Our second result characterises such lumpings by equality of natural entropic bounds with the entropy rate of the lumped process in the stationary setting. A first equality holding only for entropies depending on the lumped process is equivalent to \emph{weak lumpability}, i.e. the lumped process is a higher-order \NameMarkov{} chain in the stationary setting. A second equality involving entropies also using the underlying \NameMarkov{} chain in the stationary case is equivalent to \emph{strong lumpability}, i.e. the lumped process is a higher-order \NameMarkov{} chain, \emph{for every} initial distribution. Our characterisation is an information theoretic complement to \NameGurvitsLedoux{}'s~\cite{Gurvits_Ledoux__MarkovPropertyForAFunctionOfAMarkovChain_ALinearAlgebraApproach} linear algebraic approach to characterise lumpability.\\

We state a sufficient condition on the transition graph and the lumping function to preserve the entropy rate and be strongly $k$-lumpable. The condition is fulfilled on non-trivial lower-dimensional subspaces of the space of transition matrices. This complements \NameGurvitsLedoux{}'s~\cite{Gurvits_Ledoux__MarkovPropertyForAFunctionOfAMarkovChain_ALinearAlgebraApproach} result that lumpings having higher-order \NameMarkov{} behaviour are nowhere dense.

\section{Main results}
\label{sec:main}
\subsection{Preliminaries}
\label{sec:preliminaries}
We let $\NatNum:=\Set{1,2,\dotsc}$ and $\NatNumZero:=\Set{0,1,2,\dotsc}$. Let $[n,m]\,:=\Set{k\in\NatNumZero: n\le k \le m}$ and abbreviate $\IntegerSet{n}:=\IntegerSet{1,n}$. A vector $\Sequence{x}$ subscripted by a set $A$ is the subvector of elements indexed by this set: $\Sequence{x}_A:=(x_n)_{n\in A}$.\\

We recall information-theoretic basics from \NameCoverThomas{}~\cite[chapters 2 \& 4]{Cover_Thomas__ElementsOfInformationTheory_Ed2__Wiley_2006}. Let $\ld$ denote the binary logarithm. By continuous extension, we assume $0\ld 0=0$. The \emph{\NameShannon{} entropy} of a rv $Z$ taking values in a finite set $\SpaceZ$ is
\begin{subequations}\label{eq:entropy}
\begin{equation}\label{eq:entropy:base}
 \Entropy{Z}:=-\sum_{z\in\SpaceZ} \Proba(Z=z)\ld\Proba(Z=z)\,.
\end{equation}
The \emph{conditional entropy} of $Z$ given $W$ is defined by
\begin{equation}\label{eq:entropy:conditional}
 \Entropy{Z|W}:=\sum_{w\in\SpaceW}\Proba(W=w)\Entropy{Z|W=w}\,.
\end{equation}
Successive conditioning reduces entropy:
\begin{equation}\label{eq:entropy:reduction}
 \Entropy{Z}\ge\Entropy{Z|W_1}\ge\Entropy{Z|W_1,W_2}\,.
\end{equation}
For a stationary stochastic process $Z:=(Z_n)_{n\in\NatNumZero}$ on a finite state space $\mathcal{Z}$, the \emph{entropy rate} is
\begin{equation}\label{eq:entropy:rate}
 \EntropyRate{Z}
 :=\LimN\OneOverN\Entropy{Z_{\IntegerSet{n}}}
 =\LimN\Entropy{Z_n|Z_{[n-1]}}\,.
\end{equation}
The left limit in~\eqref{eq:entropy:rate} is the limit of the normalised \emph{block entropy} $\Entropy{Z_{\IntegerSet{n}}}$. By stationarity and~\eqref{eq:entropy:reduction}, the $\Entropy{Z_n|Z_{[n-1]}}$ in the right limit of~\eqref{eq:entropy:rate} are monotonically decreasing.
\end{subequations}
\subsection{Setting}
\label{sec:setting}

Let $X:=(X_n)_{n\in\NatNumZero}$ be an irreducible, aperiodic, time-homogeneous \NameMarkov{} chain on the finite state space $\SpaceX$. It has transition matrix $\TransitionMatrix$ with invariant probability measure $\MeasureX$. We assume that $X$ is stationary, that is, $X_0\DistributedAs\MeasureX$. The \emph{lumping function} $\LumpingFunction$ is, $\SpaceX\to\SpaceY$ and surjective. We assume $\LumpingFunction$ to be non-trivial, that is, $2\le\Cardinality{\SpaceY}<\Cardinality{\SpaceX}$. Without loss of generality, we extend $\LumpingFunction$ to $\SpaceX^n\to\SpaceY^n$ coordinate-wise, for arbitrary $n\in\NatNum$. The \emph{lumped process} of $X$ under $\LumpingFunction$ is the stationary stochastic process $Y:=(Y_n)_{n\in\NatNumZero}$ defined by $Y_n:=\LumpingFunction(X_n)$. We refer to this setup as the \emph{lumping} $\Lumping$.\\

The lumping induces a \emph{conditional entropy rate}~\cite{Geiger_Kubin__SomeResultsOnTheInformationLossInDynamicalSystems__ISWSC_2011,Watanabe_Abraham__LossAndRecoveryOfInformationByCoarseObservationOfStochasticChain__IC_19}, which characterises the average information loss per time step:
\begin{equation}\label{eq:conditionalEntropyRate}
 \EntropyRate{X|Y}
 :=\LimN\OneOverN\Entropy{X_{\IntegerSet{n}}|Y_{\IntegerSet{n}}}
 =\EntropyRate{X}-\EntropyRate{Y}\,.
\end{equation}
Our main question is whether $\EntropyRate{X|Y}$ is positive or zero, speaking of \emph{entropy rate loss} or \emph{entropy rate preservation} respectively. Entropy rate preservation does not imply that we can reconstruct the original process from the lumped process without entropy loss. See figure~\ref{fig:SFSnotStrong} (page~\pageref{fig:SFSnotStrong}) for an example.\\

The \emph{transition graph} $\TransitionGraph$ of the Markov chain $X$ is the directed graph with vertex set $\SpaceX$ and an edge $(x,x')$, iff $\Proba(X_1=x'|X_0=x)>0$. A length $n$ trajectory $\Sequence{x}\in\SpaceX^{n}$ is \emph{realisable}, iff $\Proba(X_{\IntegerSet{n}}=\Sequence{x})>0$, equivalent to being a directed path in $\TransitionGraph$. A key structural property of $\TransitionGraph$ is its \emph{split-merge index} with respect to $\LumpingFunction$:
\begin{equation}\label{eq:kappa}
 \Kappa:=
 \inf\BigSet{
  n\in\NatNum\,
  \Bigg|
  \begin{gathered}
  \Exists\check{x},\hat{x}\in\SpaceX,
  \Sequence{y}\in \SpaceY^n:
  \Exists\Sequence{x'},\Sequence{x''}\in \LumpingPreimage(\Sequence{y}),\Sequence{x''}\neq\Sequence{x'}:
  \\
  \text{ st both }
  \begin{cases}
   \Proba(X_0=\check{x},X_{\IntegerSet{n}}=\Sequence{x'},X_{n+1}=\hat{x})>0
   \\
   \Proba(X_0=\check{x},X_{\IntegerSet{n}}=\Sequence{x''},X_{n+1}=\hat{x})>0
  \end{cases}
  \end{gathered}
 }\,.
\end{equation}
The split-merge index is the shortest length of the differing part of a pair of finite, different and realisable trajectories with common start and end points, and same lumped image, if such a pair exists. Otherwise, let $\Kappa=\infty$. If $\Kappa<\infty$, every pair of sequences $\Sequence{x'},\Sequence{x''}\in\SpaceX^\Kappa$ fulfilling~\eqref{eq:kappa} is not only different, but differs in every coordinate by virtue of the infimum in~\eqref{eq:kappa}. Figure~\ref{fig:KappaFig} (page~\pageref{fig:KappaFig}) gives an example of $\Kappa\le 3$.\\

\begin{figure}
\begin{center}
%
%

\begin{tikzpicture}[scale=2,>=stealth,rounded corners]
  \drawnode{a}{0,0.5}{$\check{x}$}
  \drawnode{b1}{1,0.75}{$x'_1$}  \drawnode{b2}{1,0.25}{$x''_1$}
  \drawnode{c1}{2,0.75}{$x'_2$}  \drawnode{c2}{2,0.25}{$x''_2$}
  \drawnode{d1}{3,0.75}{$x'_3$}  \drawnode{d2}{3,0.25}{$x''_3$} 
  \drawnode{f}{4,0.5}{$\hat{x}$}

  \draw[->] (a) -- (b1);  \draw[->] (a) -- (b2);
  \draw[->] (b1) -- (c1);  \draw[->] (b2) -- (c2);
  \draw[->] (c1) -- (d1);  \draw[->] (c2) -- (d2);
  \draw[->] (d1) -- (f);  
  \draw[->] (d2) -- (f);
  \draw[->] (f) -- (4.5,0.5); 
  \draw[->] (-0.5,0.5) -- (a);

  \draw[draw=red] (-0.25,0.25) rectangle (0.25,0.75);
  \draw[draw=red] (0.75,-0.0) rectangle (1.25,1.0);
  \draw[draw=red] (2.75,-0.0) rectangle (3.25,1.0);
  \draw[draw=red] (1.75,-0.0) rectangle (2.25,1.0);
  \draw[draw=red] (3.75,0.25)  rectangle (4.25,0.75);
  \draw (1,-0.25) node {\textcolor{red}{$y_1$}};
  \draw (2,-0.25) node {\textcolor{red}{$y_2$}};
  \draw (3,-0.25) node {\textcolor{red}{$y_3$}};
  \draw (0,-0.0) node {\textcolor{red}{$g(\check{x})$}};
  \draw (4,-0.0) node {\textcolor{red}{$g(\hat{x})$}};
\end{tikzpicture}
\end{center}
\caption[A split-merge situation]{(Colour online) A section of trajectory space, with time running left-to-right. The two realisable length $5$ trajectories $(\check{x},x'_1,x'_2,x'_3,\hat{x})$ and $(\check{x},x''_1,x''_2,x''_3,\hat{x})$ have the same lumped image $(\LumpingFunction(\check{x}),y_1,y_2,y_3,\LumpingFunction(\hat{x}))$. Thus $\Kappa\le 3$. The lumped states $\Set{\LumpingFunction(\check{x}),y_1,y_2,y_3,\LumpingFunction(\hat{x})}$ need not be distinct; e.g., it might be that $y_1=y_2=g(\hat{x})$. If $\Kappa=3$, then the minimality of $\Kappa$ implies that $x'_i\not=x''_i$, for $i\in\IntegerSet{3}$.}
\label{fig:KappaFig}
\end{figure}
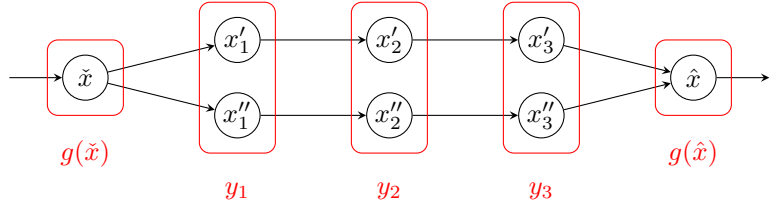

\subsection{Characterisation of entropy rate loss}
\label{sec:entropycharacterisation}
This section presents the characterisation of the entropy rate loss of a lumping in terms of $\Kappa$ and the growth rate of the cardinality of the realisable preimage. The \emph{realisable preimage} of a lumped trajectory $\Sequence{y}\in\SpaceY^n$ are the realisable trajectories in its preimage:
\begin{equation}\label{eq:realisablePreimage}
 R(\Sequence{y}):=\Set{
  \Sequence{x}\in\LumpingPreimage(\Sequence{y}):
  \Sequence{x}\text{ is realisable}
 }\,.
\end{equation}
The \emph{preimage count of length $n$} of the lumping $\Lumping$ is the cardinality of the \emph{realisable preimage} of a random lumped trajectory of length $n$:
\begin{equation}\label{eq:preimageCount}
 T_n:=\Cardinality{R(Y_{\IntegerSet{n}})}
  = \sum_{\Sequence{x}\in \LumpingPreimage(Y_{\IntegerSet{n}})}
  \Iverson{\Proba(X_{\IntegerSet{n}}=\Sequence{x})>0}\,,
\end{equation}
where the right side sums over \NameIverson{} brackets. Our first main result characterises entropy rate preservation:
\begin{Thm}\label{thm:entropycharacterisation}
\begin{subequations}\label{eq:entropycharacterisation}
\begin{align}
 \label{eq:entropycharacterisation:loss}
 \EntropyRate{X|Y}>0
 \,\Iff\,
 \Kappa<\infty
 &\,\Iff\,
 \Exists C>1:
 &&\Proba\left(\liminf_{n\to\infty}\sqrt[n]{T_n}\ge C\right)=1
 \,,
 \\
 \label{eq:entropycharacterisation:preservation}
 \EntropyRate{X|Y}=0
 \,\Iff\,
 \Kappa=\infty
 &\,\Iff\,
 \Exists C<\infty:
 && \Proba\left(\sup_{n\to\infty} T_n \le C\right)=1
 \,.
\end{align}
\end{subequations}
\end{Thm}

The proofs of all statements in this section are in section~\ref{sec:proof:entropyrate}. The constants $C$ in theorem~\ref{thm:entropycharacterisation} are explicit functions of $\Lumping$; see~\eqref{eq:trajectory:rate} for~\eqref{eq:entropycharacterisation:loss} and~\eqref{eq:trajectory:bound:explicit} for~\eqref{eq:entropycharacterisation:preservation}. Likewise, an explicit lower bound for the entropy rate loss in case~\eqref{eq:entropycharacterisation:loss} is stated in~\eqref{eq:loss:quantification}, implying that the entropy loss grows at least linearly in the sequence length.\\

Theorem~\ref{thm:entropycharacterisation} reveals a dichotomy in behaviour of the entropy of the lumping. If $\Kappa$ is infinite, then no split-merge situations as in figure~\ref{fig:KappaFig} (page~\pageref{fig:KappaFig}) occur. Thus, all finite trajectories of $X$ can be reconstructed from its lumped image and knowledge of its endpoints. Therefore, the only entropy loss occurs at those endpoints and is finite. This yields uniform finite bounds on the conditional block entropies and the preimage count. If $\Kappa$ is finite, then at least two different, realisable length $(\Kappa+2)$ trajectories of $X$ with the same lumped image split and merge (see figure~\ref{fig:KappaFig}). Such a split-merge leads to a finite entropy loss. The ergodic theorem ensures that this situation occurs linearly often in the block length, thus leading to a linear growth of the conditional block entropy. This implies an entropy rate loss. In particular, the conditional block entropy of a lumping never exhibits sublinear and unbounded growth.\\

If no split-merge situation occurs, then realisable trajectories with the same lumped image must be parallel. This constraint bounds their number. First, this yields a uniform bound on the conditional block entropies for lengths smaller than $\Kappa$:

\begin{Prop}\label{prop:blockEntropyBound}
We have
\begin{equation}\label{eq:blockEntropyBound}
 \ForAll n:\quad
 n-2<\Kappa
 \Then
 \Entropy{X_{\IntegerSet{n}}|Y_{\IntegerSet{n}}}\le
 2\ld(\Cardinality{\SpaceX}-\Cardinality{\SpaceY}+1)\,.
\end{equation}
\end{Prop}

Second, the finiteness of $\SpaceX$ implies that either a split-merge situation of low trajectory length exists or no split-merge situation exists at all:

\begin{Prop}\label{prop:kappaBound}
In case~\eqref{eq:entropycharacterisation:loss}, we have
\begin{equation}\label{eq:kappaBound}
 \Kappa\le\sum_{y\in\SpaceY}
  \Cardinality{\LumpingPreimage(y)}(\Cardinality{\LumpingPreimage(y)}-1)\,.
\end{equation}
\end{Prop}

If $P$ is positive, i.e., all its entries are positive, then $G$ is the complete directed graph and $\Kappa=1$. Hence,

\begin{Cor}\label{cor:positiveMatrix}
If $P$ is positive, then $\EntropyRate{X|Y}>0$.
\end{Cor}

Thus, entropy rate preserving lumpings must have sufficiently sparse transition matrices $P$. The examples depicted in figure~\ref{fig:LosslessNotWeak} (page~\pageref{fig:LosslessNotWeak}) and figure~\ref{fig:SESXCounterExample} (page~\pageref{fig:SESXCounterExample}) preserve the entropy rate without satisfying the sufficient conditions from section~\ref{sec:sufficient}.

\begin{figure}
\begin{center}
\begin{tikzpicture}[scale=2,>=stealth,rounded corners]
  \drawnode{a1}{1,1}{$a_1$}  \drawnode{a2}{1,0}{$a_2$}
  \drawnode{b1}{2,1}{$b_1$}  \drawnode{b2}{2,0}{$b_2$}
  \drawnode{c1}{3,1}{$c_1$}  \drawnode{c2}{3,0}{$c_2$}

  \draw[<->] (a1) -- (b1); \draw[<->] (a2) -- (b2); 
  \draw[<->] (c1) -- (b1); \draw[<->] (c2) -- (b2); 
  \draw[<->] (b1) -- (b2);
  \draw[->] (b1) -- (a2);
  \draw[->] (c1) -- (b2);
  \selfedge{a1}{180};\selfedge{a2}{180};\selfedge{c1}{0};\selfedge{c2}{0};

  \draw[draw=red] (0.75,-0.25) rectangle (1.25,1.25);
  \draw[draw=red] (2.75,-0.25) rectangle (3.25,1.25);
  \draw[draw=red] (1.75,0.75) rectangle (2.25,1.25);
  \draw[draw=red] (1.75,-0.25) rectangle (2.25,0.25);
  \draw (1,0.5) node {\textcolor{red}{$A$}}; \draw (3,0.5) node {\textcolor{red}{$C$}};
  \draw (2.4,1.25) node {\textcolor{red}{$B_1$}};
  \draw (2.4,-0.25) node {\textcolor{red}{$B_2$}};
\end{tikzpicture}
\end{center}
\caption[An entropy rate preserving lumping]{
(Colour online) The transition graph of a \NameMarkov{} chain with the lumping represented by red boxes. The lumping preserves the entropy rate without satisfying $\ClassSingleEntry$ from section~\ref{sec:sufficient}. The loops at $a_1$ and $a_2$ on the lhs, and at $c_1$ and $c_2$ on the rhs, prevent that the lumped process is $\ClassKMarkov$, for every $k$, given that the loop probabilities are different.}
\label{fig:LosslessNotWeak}
\end{figure}
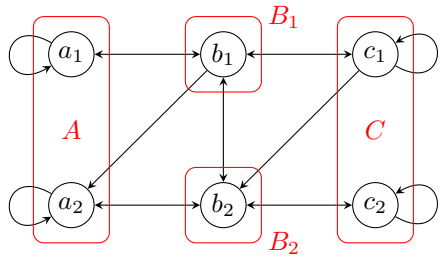

\subsection{Characterisation of strong \texorpdfstring{$k$}{k}-lumpability}
\label{sec:lumpability}
The case of the lumped process retaining the \NameMarkov{} property is desirable from a computational and modelling point of view. However, in general, the lumped process $Y$ does not possess the \NameMarkov{} property~\cite{Kemeny_Snell__FiniteMarkovChains__Springer_1976,Gurvits_Ledoux__MarkovPropertyForAFunctionOfAMarkovChain_ALinearAlgebraApproach}. Nevertheless, one may hope that the lumped process belongs to the larger and still desirable class of higher-order \NameMarkov{} chains.

\begin{Def}\label{def:kthOrderMC}
A stochastic process $Z:=(Z_n)_{n\in\NatNumZero}$ is a $k$-th order homogeneous \NameMarkov{} chain (short: $Z\text{ is }\ClassKMarkov$), iff
\begin{multline}\label{eq:markov:kthorder}
 \ForAll n\in\NatNum
 ,m\in[k,n]
 ,z_n\in \SpaceZ
 ,\Sequence{z}\in\SpaceZ^m
 :\quad
 \Proba(Z_{[n-m,n-1]}=\Sequence{z})>0 \Then
 \\
  \Proba(Z_n=z_n|Z_{[n-m,n-1]}=\Sequence{z})
  =\Proba(Z_n=z_n|Z_{[n-k,n-1]}=\Sequence{z}_{[n-k,n-1]})\,.
\end{multline}
\end{Def}

The entropy rate of a $\ClassKMarkov$ is as straightforward as one would expect:

\begin{Prop}\label{prop:ratetoMarkov}
Let $Z:=(Z_n)_{n\in\NatNumZero}$ be a stationary stochastic process on $\SpaceZ$. Then
\begin{equation}\label{eq:rateToMarkov}
 Z\text{ is }\ClassKMarkov
 \Iff
 \EntropyRate{Z}=\Entropy{Z_k|Z_{\IntegerSet{0,k-1}}}\,.
\end{equation}
\end{Prop}
The proof of this proposition is in section~\ref{sec_proof_lumpability}. We investigate lumpings, where the lumped process is $\ClassKMarkov$:

\begin{Def}[Extension of~{\cite[Def.~6.3.1]{Kemeny_Snell__FiniteMarkovChains__Springer_1976}}]\label{def:klump}
A lumping $\Lumping$ of a stationary \NameMarkov{} chain is \emph{weakly $k$-lumpable}, iff $Y$ is $\ClassKMarkov$. It is \emph{strongly $k$-lumpable}, iff this holds for each distribution of $X_0$ and the transition probabilities of $Y$ are independent of this distribution.
\end{Def}

A direct expression of the entropy rate of the lumped process $Y$ is intrinsically complicated~\cite{Blackwell__TheEntropyOfFunctionsOfFiniteStateMarkovChains__CAS_1957}. See section~\ref{sec_blackwell}. However, there are asymptotically tight, monotone decreasing, upper and lower bounds:

\begin{Lem}[{\cite[Thm.~4.5.1,~pp.~86]{Cover_Thomas__ElementsOfInformationTheory_Ed2__Wiley_2006}}]
\label{lem:ratebounds}
In our setup, we have:
\begin{equation}\label{eq_ratebounds}
 \ForAll n\in\NatNum:\quad
 \Entropy{Y_n|Y_{\IntegerSet{n-1}},X_0}
 \leq \EntropyRate{Y}
 \leq \Entropy{Y_n|Y_{\IntegerSet{0,n-1}}}\,.
\end{equation}
\end{Lem}

In the stationary setting, equality on the rhs in~\eqref{eq_ratebounds}, for $n=k$, together with proposition~\ref{prop:ratetoMarkov} implies that $Y$ is $\ClassKMarkov$, i.e. $\Lumping$ is weakly $k$-lumpable. If there is also equality on the lhs in~\eqref{eq_ratebounds}, for $n=k$, then knowledge of the distribution of $X_0$ delivers no additional information about $Y_k$. In other words, $Y$ is $\ClassKMarkov$, for every starting distribution. Our second main results characterises higher-order lumpability:

\begin{Thm}\label{thm:klump:equalities}
The following statements are equivalent:
\begin{subequations}\label{eq:klump}
\begin{equation}\label{eq:klump:boundsequal}
 \Entropy{Y_k|Y_{\IntegerSet{k-1}},X_0}=\Entropy{Y_k|Y_{\IntegerSet{0,k-1}}}\,,
\end{equation}
\begin{equation}\label{eq:klump:strong}
 X \text{ is strongly $k$-lumpable.}
\end{equation}
\end{subequations}
\end{Thm}
The proof of theorem~\ref{thm:klump:equalities} is in section~\ref{sec_proof_lumpability}. We stress the fact that~\eqref{eq:klump:boundsequal} is a condition only on the stationary setting, whereas~\eqref{eq:klump:strong} deals with all starting distributions. Theorem~\ref{thm:klump:equalities} is an information theoretic equivalent to \NameGurvitsLedoux{}'s characterisation~\cite[theorems~2~and~6]{Gurvits_Ledoux__MarkovPropertyForAFunctionOfAMarkovChain_ALinearAlgebraApproach} of $k$-lumpability via a linear algebraic description of invariant subspaces. A classic example~\cite[pp.~139]{Kemeny_Snell__FiniteMarkovChains__Springer_1976} shows that weak $k$-lumpability alone is not sufficient for~\eqref{eq:klump}. Moreover, the examples in figures~\ref{fig:SESXCounterExample} (page~\pageref{fig:SESXCounterExample}) and~\ref{fig:SFSnotStrong} (page~\pageref{fig:SFSnotStrong}) and example~\ref{exam:strong2notSFS2} (page~\pageref{exam:strong2notSFS2}) are strongly lumpable for some $k$ without satisfying the sufficient condition from section~\ref{sec:sufficient}.

\subsection{Sufficient conditions}
\label{sec:sufficient}
We present easy-to-check sufficient conditions for the preservation of the entropy rate and strong $k$-lumpability. Their proofs are in section~\ref{sec:proof:sufficient}. The conditions depend only on the transition graph $G$ and the lumping function $\LumpingFunction$.\\ 

Our first sufficient condition preserves the entropy rate:
\begin{Def}\label{def:singleEntry}
A lumping $\Lumping$ is \emph{single entry} (short: $\ClassSingleEntry$), iff
\begin{multline}\label{eq:singleEntry}
 \ForAll y\in\SpaceY,x\in\SpaceX:
 \Exists x'\in \LumpingPreimage(y):
 \ForAll x''\in \LumpingPreimage(y)\setminus\Set{x'}:
 \\
 \Proba(X_1=x''|X_0=x)=0\,,
\end{multline}
i.e., there is \emph{at most one} edge from a given state $x$ into the preimage $\LumpingPreimage(y)$.
\end{Def}

The $\ClassSingleEntry$ lumpings are entropy rate preserving:
\begin{Prop}\label{prop:SESXpreserves}
If $\Lumping$ is $\ClassSingleEntry$, then $\EntropyRate{X|Y}=0$.
\end{Prop}
Figure~\ref{fig:LosslessNotWeak} (page~\pageref{fig:LosslessNotWeak}) and figure~\ref{fig:SESXCounterExample} (page~\pageref{fig:SESXCounterExample}) show that $\ClassSingleEntry$ is not necessary for entropy rate preservation.\\

\begin{Cor}\label{cor:SEweakStrong}
If $\Lumping$ is $\ClassSingleEntry$ and weakly $k$-lumpable, then it is strongly $k$-lumpable.
\end{Cor}

\begin{proof}
The proof of proposition~\ref{prop:SESXpreserves} shows that $\ClassSingleEntry$ implies equality on the lhs of~\eqref{eq_ratebounds}, for all $n$. Weak $k$-lumpability implies equality on the rhs of~\eqref{eq_ratebounds}, for $n=k$. Therefore, theorem~\ref{thm:klump:equalities} applies.
\end{proof}

An example of a lumping satisfying the conditions of the corollary is given in figure~\ref{fig:SFSnotStrong} (page~\pageref{fig:SFSnotStrong}). That a lumping can be $\ClassSingleEntry$ without being strongly lumpable, or strongly lumpable without being $\ClassSingleEntry$ is shown in figure~\ref{fig:SESFSCounterExample} (page~\pageref{fig:SESFSCounterExample}) and in example~\ref{exam:strong2notSFS2} (page~\pageref{exam:strong2notSFS2}) respectively.\\

\begin{figure}
\begin{center}
\begin{tikzpicture}[scale=2,>=stealth,rounded corners]
  \drawnode{a}{0,0.5}{$a$}
  \drawnode{b1}{1,1}{$b_1$}  \drawnode{b2}{1,0}{$b_2$}
  \drawnode{c1}{2,1}{$c_1$}  \drawnode{c2}{2,0}{$c_2$}

  \draw[->] (a) -- (b1);  \draw[->] (a) -- (b2);
  \draw[->] (b1) -- (c1);  \draw[->] (b2) -- (c2);
  \draw[->] (c1) .. controls (2,2) and (0,2) .. (a);
  \draw[->] (c2) .. controls (2,-1) and (0,-1) .. (a);
  \selfedge{a}{180};

  \draw[draw=red] (0.75,-0.25) rectangle (1.25,1.25);
  \draw[draw=red] (1.75,0.75) rectangle (2.25,1.25);
  \draw[draw=red] (1.75,-0.25) rectangle (2.25,0.25);
  \draw[draw=red] (-0.25,0.25) rectangle (0.25,0.75);
  \draw (1,0.5) node {\textcolor{red}{$B$}};
  \draw (-0.4,0.75) node {\textcolor{red}{$A$}};
  \draw (2.4,1.25) node {\textcolor{red}{$C_1$}};
  \draw (2.4,0.25) node {\textcolor{red}{$C_2$}};
\end{tikzpicture}
\end{center}
\caption[$\ClassSingleEntry$ is neither necessary for entropy rate preservation nor for weak $k$-lumpability]{(Colour online) The transition graph of a Markov chain with the lumping represented by red boxes. The lumping is not $\ClassSingleEntry$ (violated by transitions from $a$ into $B$). On the other hand, the existence of the uniquely represented states $C_1$ and $C_2$ allows to distinguish between the trajectories $(a,b_1,c_1,a)$ and $(a,b_2,c_2,a)$. Therefore, the lumping preserves the entropy rate. Furthermore, this lumping is weakly $1$-lumpable and strongly $2$-lumpable, but not strongly $1$-lumpable. Hence it shows that $\ClassSingleEntry$ is neither necessary for entropy rate preservation nor for weak $k$-lumpability. This also applies to $\ClassSFS{k}$, a subclass of $\ClassSingleEntry$.}
\label{fig:SESXCounterExample}
\end{figure}
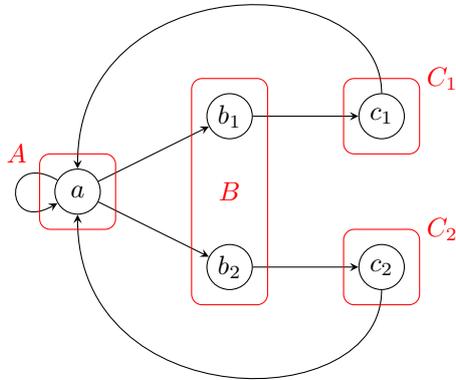

Our second sufficient condition preserves the entropy rate and guarantees higher-order lumpability:
\begin{Def}\label{def:singleforwardsequence}
For $k\ge 2$, a lumping $\Lumping$ has the \emph{single forward $k$-sequence property} (short: $\ClassSFS{k}$), iff
\begin{multline}\label{eq:singleforwardsequence}
 \ForAll\Sequence{y}\in\SpaceY^{k-1},y\in\SpaceY:
 \Exists \Sequence{x'}\in \LumpingPreimage(\Sequence{y}):
 \ForAll x\in \LumpingPreimage(y),\Sequence{x}\in \LumpingPreimage(\Sequence{y})\setminus\Set{\Sequence{x'}}:
 \\\Proba(
  X_{\IntegerSet{k-1}}=\Sequence{x}
  |Y_{\IntegerSet{k-1}}=\Sequence{y},X_0=x
 )=0\,,
\end{multline}
i.e., there is \emph{at most one} realisable sequence in the preimage $\LumpingPreimage(\Sequence{y})$ starting in $y$.
\end{Def}

The $\ClassSFS{k}$ property implies entropy rate preservation and strong $k$-lumpability:
\begin{Prop}\label{prop:UEimpliesklump}
If $\Lumping$ is $\ClassSFS{k}$, then it is strongly $k$-lumpable and $\ClassSingleEntry$.
\end{Prop}

Figure~\ref{fig:LossLumpability} (page~\pageref{fig:LossLumpability}) gives an overview of the various classes and examples, in particular showing that the sufficient conditions are not necessary. Figures~\ref{fig:SESXCounterExample} (page~\pageref{fig:SESXCounterExample}) and~\ref{fig:SFSnotStrong} (page~\pageref{fig:SFSnotStrong}) show that $\ClassSFS{2}$ is neither necessary for weak $1$-lumpability, nor for entropy rate preservation, nor for $\ClassSingleEntry$. Figure~\ref{fig:SESFSCounterExample} (page~\pageref{fig:SESFSCounterExample}) shows that $\ClassSingleEntry$ does neither imply $\ClassSFS{k}$ nor strong $k$-lumpability, for every $k$. Figure~\ref{fig:SFS2Example} (page~\pageref{fig:SFS2Example}) gives an example of a lumping being $\ClassSFS{2}$ and not strongly $1$-lumpable. Finally, example~\ref{exam:strong2notSFS2} (page~\pageref{exam:strong2notSFS2}) gives a strongly $2$-lumpable lumping which is not $\ClassSFS{2}$.

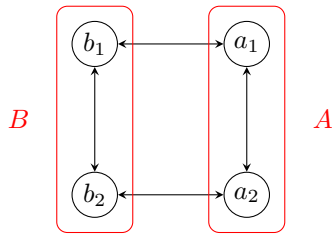
\begin{figure}
\begin{center}
\begin{tikzpicture}[scale=2,>=stealth,rounded corners]
  \drawnode{b1}{1,1}{$b_1$}  \drawnode{b2}{1,0}{$b_2$}
  \drawnode{a}{2,1}{$a_1$}  \drawnode{c}{2,0}{$a_2$}

  \draw[<->] (a) -- (b1);
  \draw[<->] (b1) -- (b2); 
  \draw[<->] (a) -- (c);
  \draw[<->] (c) -- (b2);

  \draw[draw=red] (0.75,-0.25) rectangle (1.25,1.25);
  \draw[draw=red] (1.75,-0.25) rectangle (2.25,1.25);
  \draw (0.5,0.5) node {\textcolor{red}{$B$}};
  \draw (2.5,0.5) node {\textcolor{red}{$A$}};
\end{tikzpicture}
\end{center}
\caption[$\ClassSFS{k}$ is neither necessary for entropy rate preservation nor for strong lumpability]{
(Colour online) The transition graph of a \NameMarkov{} chain with the lumping represented by red boxes. The lumping is $\ClassSingleEntry$ and thus preserves the entropy rate. Furthermore, if all transitions have probability $1/2$, it is strongly $1$-lumpable and thus $\Entropy{Y_1|X_0}=\Entropy{Y_1|Y_0}$ (see theorem~\ref{thm:klump:equalities}). However, observing an arbitrarily long trajectory of the lumped process does not determine the current preimage state. Whence $\Lumping$ is not $\ClassSFS{k}$, for every $k$. Therefore, $\ClassSFS{k}$ is neither necessary for entropy rate preservation nor for strong lumpability.}
\label{fig:SFSnotStrong}
\end{figure}

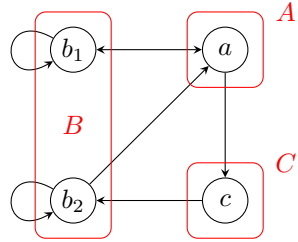
\begin{figure}
\begin{center}
\begin{tikzpicture}[scale=2,>=stealth,rounded corners]
  \drawnode{b1}{1,1}{$b_1$}  \drawnode{b2}{1,0}{$b_2$}
  \drawnode{a}{2,1}{$a$}  \drawnode{c}{2,0}{$c$}

  \draw[<->] (a) -- (b1); 
  \draw[->] (a) -- (c);
  \draw[->] (c) -- (b2);
  \draw[->] (b2) -- (a);
  \selfedge{b1}{180};\selfedge{b2}{180};

  \draw[draw=red] (0.75,-0.25) rectangle (1.25,1.25);
  \draw[draw=red] (1.75,0.75) rectangle (2.25,1.25);
  \draw[draw=red] (1.75,-0.25) rectangle (2.25,0.25);
  \draw (1,0.5) node {\textcolor{red}{$B$}};
  \draw (2.4,1.25) node {\textcolor{red}{$A$}};
  \draw (2.4,0.25) node {\textcolor{red}{$C$}};
\end{tikzpicture}
\end{center}
\caption[$\ClassSingleEntry$ does neither imply $\ClassSFS{k}$ nor strong $k$-lumpability]{
(Colour online) The transition graph of a \NameMarkov{} chain with the lumping represented by red boxes. The lumping is $\ClassSingleEntry$. The loops at $b_1$ and $b_2$ imply that the lumped process is not $\ClassKMarkov$, for every $k$ and regardless of the distribution of $X_0$. This is easily seen by the inability to differentiate between $n$ consecutive $b_1$'s and $n$ consecutive $b_2$'s. When starting in $B$ and as long as $\Proba(X_1=a|X_0=b_1)\not=\Proba(X_1=a|X_0=b_2)$ and $\Proba(X_1=b_1|X_0=b_1)\not=\Proba(X_1=b_2|X_0=b_2)$, this long sequence of $B$s prevents determining the probability of entering $A$. Thus it is neither $\ClassSFS{k}$ nor strongly $k$-lumpable, for each $k$.
}
\label{fig:SESFSCounterExample}
\end{figure}

\begin{figure}
\begin{center}
\begin{tikzpicture}[scale=2,>=stealth,rounded corners]
  \drawnode{b1}{1,1}{$b_1$}  \drawnode{b2}{1,0}{$b_2$}
  \drawnode{c}{2,1.5}{$c$}
  \drawnode{a1}{3,1}{$a_1$}  \drawnode{a2}{3,0}{$a_2$}

  \draw[->] (a2) -- (a1); 
  \draw[->] (a2) -- (c);
  \draw[<->] (c) -- (a1);
  \draw[<->] (a2) -- (b2);
  \draw[->] (b2) -- (b1);
  \draw[->] (b1) -- (c);
  \draw[<->] (c) -- (b2);
  \selfedge{a1}{0};\selfedge{b1}{180};\selfedge{c}{90};

  \draw[draw=red] (0.75,-0.25) rectangle (1.25,1.25);
  \draw[draw=red] (2.75,-0.25) rectangle (3.25,1.25);
  \draw[draw=red] (1.75,1.25) rectangle (2.25,1.75);
  \draw (0.6,0.5) node {\textcolor{red}{$B$}};
  \draw (3.4,0.5) node {\textcolor{red}{$A$}};
  \draw (2.4,1.75) node {\textcolor{red}{$C$}};
\end{tikzpicture}
\end{center}
\caption[An $\ClassSFS{2}$ lumping]{(Colour online) The transition graph of a \NameMarkov{} chain with the lumping represented by red boxes. After at most two steps one either enters a new lumped state at a unique original state or is circling in either $b_1$ or $a_1$. Hence, this lumping is $\ClassSFS{2}$ and not strongly $1$-lumpable. The space of \NameMarkov{} chains with this transition graph contains at least the interior of a multi-simplex in $\RealNum^{13}$, parametrised by $8$ parameters ($13$ directed edges minus $5$ nodes).}
\label{fig:SFS2Example}
\end{figure}
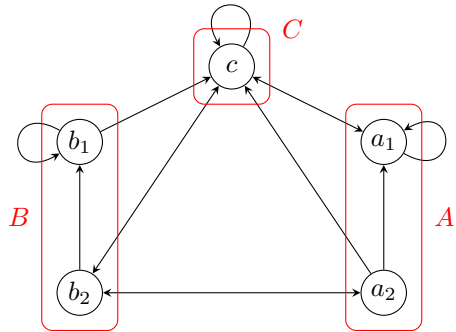

\begin{Exam}\label{exam:strong2notSFS2}
Consider the following transition matrix, where the lines divide lumped states:
\begin{equation*}
P:=
\left[
\begin{array}{cc|cc}
 0.6   &  0.4  &  0    &  0  
 \\
 0.3  &  0.2 &  0.1 &  0.4  
 \\ \hline
 0.2   &  0.05 &  0.375 &  0.375 
 \\
  0.2   &  0.05 &  0.375 &  0.375  
\end{array}
\right]\,.
\end{equation*}
This lumping is strongly $2$-lumpable and satisfies~\eqref{eq:klump:boundsequal} with $\EntropyRate{Y}=\Entropy{Y_2|Y_{[0,1]}}=\Entropy{Y_2|Y_1,X_0}=0.733$ (with an accuracy of $0.001$). However, it does not preserve entropy: $1.480=\EntropyRate{X}>\EntropyRate{Y}$, whence it is neither $\ClassSingleEntry$ nor $\ClassSFS{2}$.
\end{Exam}

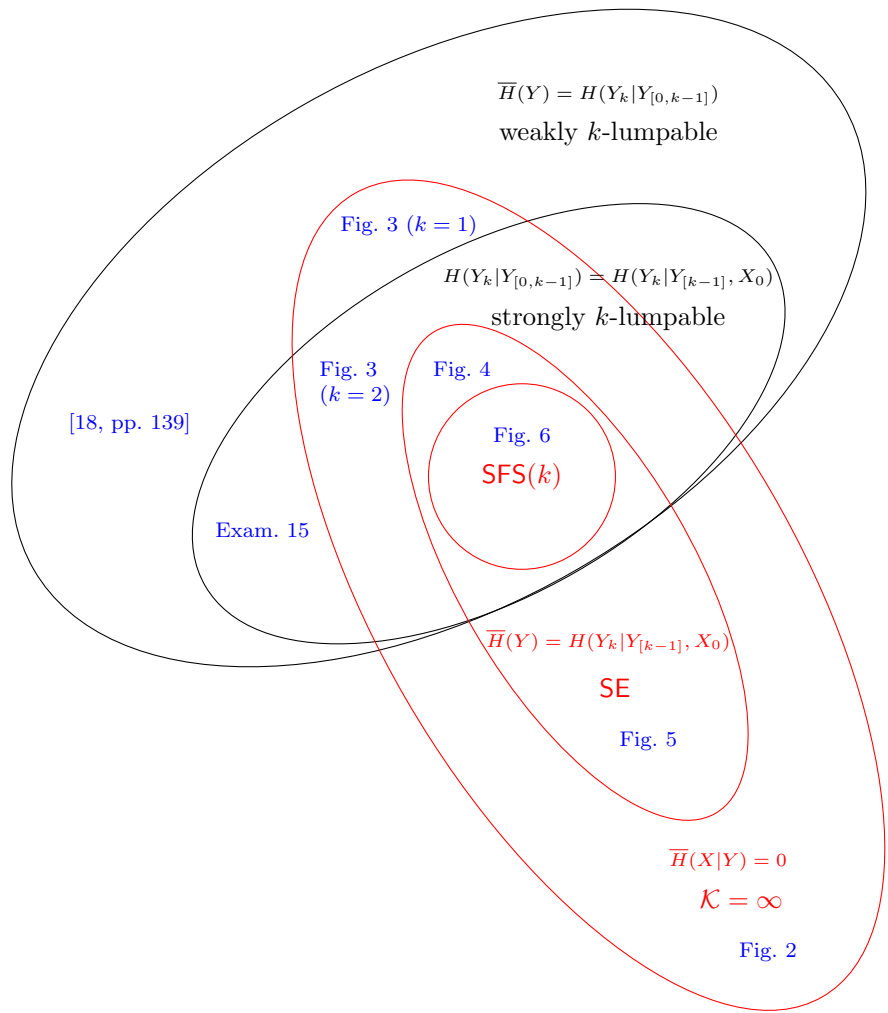
\begin{figure}
 \centering
 \begin{tikzpicture}[scale=1.75,>=stealth,rounded corners]
	\draw[rotate=30] (0,0) ellipse (70pt and 37pt);
	\draw[rotate=30] (0.0,0.75) ellipse (100pt and 58pt);

	\draw[rotate=-60,color=red] (1.3,-0) ellipse (60pt and 25pt);
	\draw[rotate=-60,color=red] (1.5,0) ellipse (100pt and 45pt);
	\draw[color=red] (0.25,-0.4) circle (20pt);

	\draw (0.25,-0.4) node {\textcolor{red}{$\ClassSFS{k}$}};

	\draw (0.95,-2) node {\textcolor{red}{$\ClassSingleEntry$}};
	\draw (0.90,-1.65) node {\scriptsize \textcolor{red}{$\EntropyRate{Y}=\Entropy{Y_k|Y_{[k-1]},X_0}$}};
	
	\draw (1.9,-3.6) node {\textcolor{red}{$\Kappa=\infty$}};
	\draw (1.8,-3.3) node {\scriptsize\textcolor{red}{$\EntropyRate{X|Y}=0$}};

	\draw (0.9,0.8) node {strongly $k$-lumpable};
	\draw (0.9,1.1) node {\scriptsize $\Entropy{Y_k|Y_{[0,k-1]}}=\Entropy{Y_k|Y_{[k-1]},X_0}$};

	\draw (0.9,2.2) node {weakly $k$-lumpable};
	\draw (0.9,2.5) node {\scriptsize $\EntropyRate{Y}=\Entropy{Y_k|Y_{[0,k-1]}}$};

	\draw (1.2,-2.4) node {\footnotesize \textcolor{blue}{Fig.~\ref{fig:SESFSCounterExample}}};
	\draw (-0.7,0.3) node {\footnotesize \textcolor{blue}{\parbox{2cm}{Fig.~\ref{fig:SESXCounterExample}\\($k=2$)}}};
	\draw (-0.6,1.5) node {\footnotesize \textcolor{blue}{Fig.~\ref{fig:SESXCounterExample}~($k=1$)}};
	\draw (0.25,-0.1) node {\footnotesize \textcolor{blue}{Fig.~\ref{fig:SFS2Example}}};
	\draw (-0.2,0.4) node {\footnotesize \textcolor{blue}{Fig.~\ref{fig:SFSnotStrong}}};
	\draw (-1.7,-0.8) node {\footnotesize \textcolor{blue}{Exam.~\ref{exam:strong2notSFS2}}};
	\draw (-2.7,0) node {\footnotesize \textcolor{blue}{{\cite[pp.~139]{Kemeny_Snell__FiniteMarkovChains__Springer_1976}}}};
	\draw (2.1,-4) node {\footnotesize \textcolor{blue}{Fig.~\ref{fig:LosslessNotWeak}}};
\end{tikzpicture}
 \caption[Relation between different classes and location of counterexamples in this paper]{(Colour online) Venn diagram of the relation between different classes and location of counterexamples in this paper.}
 \label{fig:LossLumpability}
\end{figure}

\subsection{Further discussion}
\label{sec:discussion}

The study of functions of \NameMarkov{} chains has a long tradition. In particular, whether a function of a \NameMarkov{} chain possesses the \NameMarkov{} property or not~\cite{Burke_Rosenblatt__AMarkovianFunctionOfAMarkovChain__AMS_1958,Rogers_Pitman__MarkovFunctions__AP_1981}. Kemeny \& Snell~\cite{Kemeny_Snell__FiniteMarkovChains__Springer_1976} coined the term \emph{lumpability} for retaining the \NameMarkov{} property. \NameGurvitsLedoux{}~\cite{Gurvits_Ledoux__MarkovPropertyForAFunctionOfAMarkovChain_ALinearAlgebraApproach} analysed higher-order lumpability, as we use it in this work. They showed that the class of \NameMarkov{} chains being lumpable is nowhere dense.\\

A related problem is the \emph{identification problem}, initially posed by Blackwell \& Koopmanns~\cite{Blackwell_Koopmans__OnTheIdentifiabilityProblemForFunctionsOfFiniteMarkovChains__AMS_1957}: given a stationary process on a finite state space, is it representable by a lumping of a \NameMarkov{} chain? The question of existence of a finite state space representation has a long tradition~\cite{Gilbert__OnTheIdentifiabilityProblemForFunctionsOfFiniteMarkovChains__AMS_1959,Heller__OnStochasticProcessesDerivedFromMarkovChains__AMS_1965,Anderson__TheRealizationProblemForHiddenMarkovModels__MCSS_1999}, without a definite algorithmic solution. Two results from research into this topic have a connection to the present work.\\

First, \NameCarlyle{}~\cite{Carlyle__IdentificationOfStateCalculableFunctionsOfFiniteMarkovChains__AMS_1967} shows that every stationary stochastic process on a finite state space is representable as a lumping of a \NameMarkov{} chain on an at most countable state space. The representation is $\ClassSingleEntry$. If it involves a \NameMarkov{} chain on a finite state space, then proposition~\ref{prop:SESXpreserves} guarantees entropy rate preservation of the representation.\\

Second, Gilbert~\cite{Gilbert__OnTheIdentifiabilityProblemForFunctionsOfFiniteMarkovChains__AMS_1959} shows that the distribution of a lumping of a finite-state \NameMarkov{} chain is uniquely determined by the distribution of $m$ consecutive samples, where $m$ depends on the cardinalities of the input and output alphabet. This does not contradict the nowhere dense result of \NameGurvitsLedoux{}, however, since the construction of the process distribution is different from a product of conditional distributions (as it is in the case of lumpability).\\

Moreover, the nowhere dense property does not prevent our results from being practically relevant. In particular, our sufficient condition holds for non-trivial lower-dimensional subspaces of the space of \NameMarkov{} transition matrices. See figure~\ref{fig:SFS2Example} (page~\pageref{fig:SFS2Example}). In other words, if the transition matrix is sufficiently sparse, one can hope that the lumping satisfies some of our sufficient conditions. More generally, one can hope that for a given \NameMarkov{} model there exists a lumping function with a desired output alphabet size such that the resulting lumping satisfies our sufficient conditions. Sparse transition matrices appear, e.g., in $n$-gram models in automatic speech recognition~\cite[Table~1]{Brown_deSouza_Mercer_Pietra_Lai__ClassBasedNGramModelsOfNaturalLanguage__CL_1992}, chemical reaction networks~\cite{Henzinger_Mikeev_Mateescu_Wolf__HybridNumericalSolutionOfTheChemicalMasterEquation__CMSB_2010,Heiner_Rohr_Schwarick_Streif__AComparativeStudyOfStochasticAnalysisTechniques__CMSB_2010,Wilkinson_StochasticModellingForSystemsBiology__TF_2011} and link prediction and path analysis~\cite{Sarukkai__LinkPredictionAndPathAnalysisUsingMarkovChains__CN_2000}. That the sufficient conditions for entropy preservation and weak $k$-lumpability are not
overly restrictive was recently shown for a letter bi-gram model~\cite{Geiger_Temmel__InformationPreservingMarkovAggregation__ITW_2013}: The bi-gram model exhibited the $\ClassSFS{2}$-property and thus permitted lossless compression.\\

In the non-stationary case, i.e. with $X_0$ having a different distribution than the invariant one, we are still \emph{stationary in the asymptotic mean}~\cite{Kieffer_Rahe__MarkovChannelsAreAsymptoticallyMeanStationary__JMA_1981,Gray__EntropyAndInformationTheory__Springer_2009}. In particular, we have entropy rates and an ergodic theorem. Hence, all statements of this paper should generalise to this setting. Whether we can drop the restriction to aperiodic and irreducible chains is a more difficult question.\\

We give crude upper bounds on the algorithmic complexity of checking the properties introduced in the present paper. By proposition~\ref{prop:kappaBound}, determining the finiteness and value of $\Kappa$ takes at most $\mathcal{O}(\Cardinality{\SpaceY}\exp(1+\Cardinality{\SpaceX}^2))$ steps. We can check the $\ClassSingleEntry$ property in $\mathcal{O}(\Cardinality{\SpaceX}^2)$ steps and the $\ClassSFS{k}$ property in $\mathcal{O}(\Cardinality{\SpaceX}^k)$ steps. Finally, the verification of strong $k$-lumpability via~\eqref{eq:klump:boundsequal} requires $\mathcal{O}(\Cardinality{\SpaceX}^{k+1})$ steps. The last bound is of a similar order as \NameGurvitsLedoux{}'s algorithm for weak $k$-lumpability~\cite[Section~2.2.2]{Gurvits_Ledoux__MarkovPropertyForAFunctionOfAMarkovChain_ALinearAlgebraApproach}. Details are in section~\ref{sec_algorithmic}.\\

There is another notion of information loss through lumping: Lindqvist~\cite{Lindqvist__OnTheLossOfInformationIncurredByLumpingStatesOfAMarkovChain__SJS_1978} discusses sufficient statistics for estimating $X_0$ from $Y_n$. \NameGurvitsLedoux{} introduced $g$-observability~\cite[Section~3]{Gurvits_Ledoux__MarkovPropertyForAFunctionOfAMarkovChain_ALinearAlgebraApproach} for determining $X_0$ from $Y_{\IntegerSet{0,n}}$. Simple examples show that entropy rate preservation is independent of $g$-observability. See section~\ref{sec_gObservability}.
\section{Proof of entropy rate preservation}
\label{sec:proof:entropyrate}

\begin{proof}[Proof of theorem~\ref{thm:entropycharacterisation}]

Statement~\eqref{eq:entropycharacterisation} follows from the mutually exhaustive implications
\begin{subequations}\label{eq:entropyrate}
\begin{align}
 \label{eq:entropyrate:loss}
 \Kappa<\infty
 &\Then \EntropyRate{X|Y}>0\,,
 \\
 \label{eq:entropyrate:preservation}
 \Kappa=\infty
 &\Then \EntropyRate{X|Y}=0
\end{align}
\end{subequations}
and
\begin{subequations}\label{eq:trajectory}
\begin{align}
 \label{eq:trajectory:explosion}
 \Kappa<\infty
 &\Then \Exists C>1:
 \Proba\left(\liminf_{n\to\infty}\sqrt[n]{T_n}\ge C\right)=1\,,
 \\
 \label{eq:trajectory:bound}
 \Kappa=\infty
 &\Then \Exists C<\infty: \Proba\left(\sup_{n\to\infty} T_n \le C\right)=1\,.
\end{align}
\end{subequations}

The proofs of implications~\eqref{eq:entropyrate:preservation} and~\eqref{eq:trajectory:bound} and of proposition~\ref{prop:blockEntropyBound} are in section~\ref{sec:proof:preservation} and the proofs of implications~\eqref{eq:entropyrate:loss} and~\eqref{eq:trajectory:explosion} and of proposition~\ref{prop:kappaBound} are in section~\ref{sec:proof:loss}. Sections~\ref{sec:nonOverLappingTraversalInstants} and~\ref{sec:conditionalMarkovProperty} contain technical results about \NameMarkov{} chains needed in the proof of the loss case in section~\ref{sec:proof:loss}.

\end{proof}
\subsection{The preservation case}
\label{sec:proof:preservation}
The definition of $\Kappa$ in~\eqref{eq:kappa} implies that lumped trajectories of length less than $\Kappa$ have a unique preimage contingent on the endpoints, i.e., if $n<\Kappa$, then $\ForAll\check{x},\hat{x}\in\SpaceX,\Sequence{y}\in\SpaceY^n$:
\begin{multline}\label{eq:uniquePreimage}
 \Proba(X_0=\check{x},Y_{\IntegerSet{n}}=\Sequence{y},X_{n+1}=\hat{x})>0
 \\\Then
 \ExistsUnique\Sequence{x}\in\SpaceX^n:
 \Proba(
  X_{\IntegerSet{n}}=\Sequence{x}|
  X_0=\check{x},Y_{\IntegerSet{n}}=\Sequence{y},X_{n+1}=\hat{x}
 )=1\,.
\end{multline}

\begin{proof}[Proof of proposition~\ref{prop:blockEntropyBound}]
We assume $n-2<\Kappa$. The unique preimage~\eqref{eq:uniquePreimage} implies that the conditional entropy of the interior of a block, given its lumped image and the states at its ends, is zero:
\begin{multline}\label{eq:uniquePreimageEntropy}
 \Entropy{X_{[2,n-1]}|X_1,X_n,Y_{\IntegerSet{n}}}
 \\=
 \sum_{\substack{\Sequence{y}\in\SpaceY^n\\\check{x},\hat{x}\in\SpaceX}}
  \Proba(X_1=\check{x},X_n=\hat{x},Y_{\IntegerSet{n}}=\Sequence{y})
  \underbrace{
  \Entropy{X_{[2,n-1]}|X_1=\check{x},X_n=\hat{x},Y_{\IntegerSet{n}}=\Sequence{y}}
  }_{=0\text{ by~\eqref{eq:uniquePreimage}}}
 \\=0\,.
\end{multline}
We apply the chain rule of entropy (cf.~\cite[pp.~22]{Cover_Thomas__ElementsOfInformationTheory_Ed2__Wiley_2006}) to decompose the conditional block entropy into its interior and its boundary. The interior vanishes by~\eqref{eq:uniquePreimageEntropy} and the entropy at the endpoints is maximal for the uniform distribution:
\begin{align*}
 \Entropy{X_{\IntegerSet{n}}|Y_{\IntegerSet{n}}}
 &=\Entropy{X_{[2,n-1]}|X_1,X_n,Y_{\IntegerSet{n}}} + \Entropy{X_1,X_n|Y_{\IntegerSet{n}}}
 \\&\le 0 + \Entropy{X_1,X_n|Y_1,Y_n}
 \\&\le \Entropy{X_1|Y_1} + \Entropy{X_n|Y_n}
 \\&\le 2\max\Set{\ld\Cardinality{\LumpingPreimage(y)}: y\in\SpaceY}
 \\&\le 2\ld(\Cardinality{\SpaceX}-\Cardinality{\SpaceY}+1)\,.
\end{align*}
\end{proof}

\begin{proof}[Proof of~\eqref{eq:entropyrate:preservation}]
As $\Kappa=\infty$, the bound from~\eqref{eq:blockEntropyBound} holds uniformly. Thus
\begin{equation*}
 \EntropyRate{X|Y}
 =\LimN\OneOverN\Entropy{X_{\IntegerSet{n}}|Y_{\IntegerSet{n}}}
 \le\LimN\frac{2\ld(\Cardinality{\SpaceX}-\Cardinality{\SpaceY}+1)}{n}
 =0\,.
\end{equation*}
\end{proof}

\begin{proof}[Proof of~\eqref{eq:trajectory:bound}]
Recall that we assume $\Kappa=\infty$. We show that, for all $\Sequence{y}\in\SpaceY^n$ with $\Proba(Y_{\IntegerSet{n}}=\Sequence{y})>0$, we have
\begin{equation}\label{eq:trajectory:bound:explicit}
 \Proba(
  T_n\le(\Cardinality{\SpaceX}-\Cardinality{\SpaceY}+1)^2
  |Y_{\IntegerSet{n}}=\Sequence{y}
 )=1\,.
\end{equation}
This implies~\eqref{eq:trajectory:bound}. To show~\eqref{eq:trajectory:bound:explicit}, we use~\eqref{eq:uniquePreimage} to bound
\begin{align*}
 &\sum_{\Sequence{x}\in \LumpingPreimage(\Sequence{y})}
  \Iverson{\Proba(X_{\IntegerSet{n}}=\Sequence{x})>0}\,.
 \\=& \sum_{x_1,x_n\in \LumpingPreimage(\Sequence{y}_{\Set{1,n}})}
  \Iverson{\Proba(X_1=x_1,X_n=x_n|Y_{\IntegerSet{n}}=\Sequence{y})>0}
 \\ &\times\sum_{\Sequence{x}\in \LumpingPreimage(\Sequence{y}_{[2,n-1]})}
  \Iverson{\Proba(X_{[2,n-1]}=\Sequence{x}|X_1=x_1,X_n=x_n,Y_{[2,n-1]}=\Sequence{y}_{[2,n-1]})>0}
 \\\le& \sum_{x_1,x_n\in \LumpingPreimage(\Sequence{y}_{\Set{1,n}})}
  \Iverson{\Proba(X_1=x_1,X_n=x_n|Y_{\IntegerSet{n}}=\Sequence{y})>0}
 \\\le&\, \Cardinality{\LumpingPreimage(\Sequence{y}_{\Set{1,n}})}\
  \le (\Cardinality{\SpaceX}-\Cardinality{\SpaceY}+1)^2\,.
\end{align*}
\end{proof}
\subsection{Non-overlapping traversal instants}
\label{sec:nonOverLappingTraversalInstants}

The main result of this section is an almost-sure \emph{linear lower growth bound} for non-overlapping occurrences of a fixed, finite pattern in a realisation in proposition~\ref{prop:lowerBound}.\\

Let $Z:=(Z_n)_{n\in\NatNum}$ be a stationary stochastic process taking values in $\SpaceZ$. The \emph{occupation instants} of a state $z$ is the set of indices
\begin{subequations}\label{eq:instants}
\begin{equation}\label{eq:instantsOccupation}
 \InstantsOccupation{Z}{z}{n}:=\Set{i\in\IntegerSet{n}: Z_i=z}\,.
\end{equation}
The classic \emph{occupation time}~\cite[section 6.4]{Parzen__StochasticProcesses__SIAM_1999} is the cardinality of the occupation instants. The \emph{traversal instants} of a sequence $\Sequence{z}\in\SpaceZ^{k}$ is the set of indices
\begin{equation}\label{eq:instantsTraversal}
 \InstantsTraversal{Z}{\Sequence{z}}{n}:=\Set{
  i\in\IntegerSet{n-k+1}:
  Z_{\IntegerSet{i,i+k-1}}=\Sequence{z}
 }\,.
\end{equation}
The \emph{non-overlapping traversal instants} of a sequence $\Sequence{z}\in\SpaceZ^{k}$ is the set of indices
\begin{equation}\label{eq:instantsNonOverLappingTraversal}
 \InstantsNonOverLappingTraversal{Z}{\Sequence{z}}{n}
 :=\BigSet{i\in [n-k+1]:
  \begin{gathered}
   Z_{[i,i+k-1]}=\Sequence{z}\\
   \ForAll j\in [i+1,i+k-1]: Z_{[j,j+k-1]}\not=\Sequence{z}
  \end{gathered}
 }\,,
\end{equation}
\end{subequations}
where we select lower indices greedily.\\

For $k\in\NatNum$, the \emph{$k$-transition process} $Z^{(k)}$ of $Z$ is the stochastic process on $\SpaceZ^k$ with marginals
\begin{equation}\label{eq:transitionProcess}
 \Proba(Z^{(k)}_{\IntegerSet{n}}=(\Sequence{z}^{i})_{i=1}^n)
 =\Proba(
  Z_{\IntegerSet{n-1}}=(\Sequence{z}_\Set{1}^{i})_{i=1}^{n-1},
  Z_{\IntegerSet{n,n+k-1}}=\Sequence{z}^{n}
 )\,,
\end{equation}
if $\ForAll i\in\IntegerSet{n-1}: \Sequence{z}^{i}_{\IntegerSet{2,k}}=\Sequence{z}^{i+1}_{\IntegerSet{k-1}}$, and zero else.
\begin{subequations}\label{eq:instantRelations}
Obvious relations are
\begin{equation}\label{eq:instantRelation:equality}
 \InstantsTraversal{Z}{\Sequence{z}}{n}=\InstantsOccupation{Z^{(k)}}{\Sequence{z}}{n-k}
\end{equation}
and
\begin{equation}\label{eq:instantRelation:subset}
 \InstantsNonOverLappingTraversal{Z}{\Sequence{z}}{n}
 \subseteq
 \InstantsTraversal{Z}{\Sequence{z}}{n}
 \quad\text{ with }\quad
 \Cardinality{\InstantsNonOverLappingTraversal{Z}{\Sequence{z}}{n}}
 \ge\frac{1}{k}\Cardinality{\InstantsTraversal{Z}{\Sequence{z}}{n}}\,.
\end{equation}
\end{subequations}

\begin{Prop}\label{prop:lowerBound}
Let $\Sequence{s}\in\SpaceX^k$ be realisable with $p:=\Proba(X_{\IntegerSet{k}}=\Sequence{s}|X_1=\Sequence{s}_\Set{1})>0$. Then
\begin{subequations}\label{eq:lowerBound}
\begin{equation}\label{eq:lowerBound:almostSure}
 \Proba\left(
  \liminf_{n\to\infty}\OneOverN
   \Cardinality{\InstantsNonOverLappingTraversal{X}{\Sequence{s}}{n}}
  \ge \frac{p\MeasureX(\Sequence{s}_\Set{1})}{k}
 \right)=1
\end{equation}
and
\begin{equation}\label{eq:lowerBound:probability}
 \ForAll\varepsilon>0:\quad
 \LimN\Proba\left(
  \Cardinality{\InstantsNonOverLappingTraversal{X}{\Sequence{s}}{n}}
  \ge \left(\frac{p\MeasureX(\Sequence{s}_\Set{1})}{k}-\varepsilon\right)n
 \right) = 1\,.
\end{equation}
\end{subequations}
\end{Prop}

\begin{Lem}[Ergodic theorem {\cite[theorem 3.55 on page 69]{Woess__DenumerableMarkovChains__EMS_2009}}]
\label{lem:ergodicTheorem:finite}
For every homogeneous, irreducible and aperiodic \NameMarkov{} chain $Z:=(Z_n)_{n\in\NatNum}$ on a finite state space $\SpaceZ$ with invariant measure $\nu$, all $f:\SpaceZ\to\RealNum$ and each starting distribution $\alpha\in\mathcal{M}_1(\SpaceZ)$ of $Z_1$, we have
\begin{equation}\label{eq:ergodicTheorem:finite}
 \Proba_\alpha\left(
  \LimN \OneOverN \sum_{i=1}^n f(Z_i)
  = \int_\SpaceZ f(z)d\nu(z)
  =:\nu(f)
 \right)=1\,.
\end{equation}
\end{Lem}

\begin{proof}[Proof of proposition~\ref{prop:lowerBound}]
Statement~\eqref{eq:lowerBound:probability} is a direct consequence of~\eqref{eq:lowerBound:almostSure}.\\

The $k$-transition process $X^{(k)}$ of $X$ is a homogeneous \NameMarkov{} chain with transition probabilities
\begin{equation}\label{eq:transitionChainTransitionProbabilities}
 \Proba(X^{(k)}_2=\Sequence{x'}|X^{(k)}_1=\Sequence{x})
 =
 \begin{cases}
  \Proba(X_{k+1}=\Sequence{x'}_{\Set{k}}|X_{k}=\Sequence{x}_{\Set{k}})
  &\text{if }\Sequence{x}_{\IntegerSet{2,k}}
            =\Sequence{x'}_{\IntegerSet{k-1}}\,,
  \\0 &\text{else.}
 \end{cases}
\end{equation}
Furthermore, as $X$ is irreducible and aperiodic, then so is $X^{(k)}$. Its invariant measure $\MeasureX^{(k)}$ fulfils $\MeasureX^{(k)}(\Sequence{x})=\MeasureX(\Sequence{x}_\Set{1})\prod_{i=1}^{k-1} \Proba(X_2=\Sequence{x}_\Set{i+1}|X_1=\Sequence{x}_\Set{i})$.\\

Let $f$ be the indicator function of $\Sequence{s}$. We use~\eqref{eq:instantRelations} and lemma~\ref{lem:ergodicTheorem:finite} to derive
\begin{align*}
 &\FirstAlign\Proba\left(
  \liminf_{n\to\infty}\OneOverN
  \Cardinality{\InstantsNonOverLappingTraversal{X}{\Sequence{s}}{n}}
  \ge\frac{\MeasureX^{(k)}(f)}{k}
 \right)
 \\&\ge\Proba\left(
  \liminf_{n\to\infty}\OneOverN
  \Cardinality{\InstantsTraversal{X}{\Sequence{s}}{n}}
  \ge \MeasureX^{(k)}(f)
 \right)
 \\&=\Proba\left(
  \liminf_{n\to\infty}\OneOverN
  \Cardinality{\InstantsOccupation{X^{(k)}}{\Sequence{s}}{n-k}}
  \ge \MeasureX^{(k)}(f)
 \right)
 \\&=\Proba\left(
  \LimN\OneOverN
  \Cardinality{\InstantsOccupation{X^{(k)}}{\Sequence{s}}{n}}
  \ge \MeasureX^{(k)}(f)
 \right)
 \\&= 1\,.
\end{align*}
Finally, $\MeasureX^{(k)}(f)=\MeasureX^{(k)}(\Sequence{s})=p\MeasureX(\Sequence{s}_\Set{1})$.
\end{proof}

\subsection{Conditional \NameMarkov{} property}
\label{sec:conditionalMarkovProperty}
This section presents two technical statements about discrete \NameMarkov{} processes. Let $X:=(X_n)_{n\in\NatNumZero}$ be a stochastic process on the Cartesian product $\mathcal{S}:=\prod_{n\in\NatNumZero} \mathcal{S}_n$ of the finite sets $(\mathcal{S}_n)_{n\in\NatNumZero}$. For $A\subseteq\NatNumZero$, let $\mathcal{S}_A:=\prod_{n\in A} \mathcal{S}_n$. In the remainder of this section, we assume that all conditional probabilities are well-defined. The process $X$ is \NameMarkov{}, iff
\begin{multline}\label{eq:markov:classic}
 \ForAll n,m\in\NatNum
 ,m\le n
 ,s_n\in \mathcal{S}_n
 ,s_{[n-m,n-1]}\in \mathcal{S}_{[n-m,n-1]}
 :\\\Proba(X_n=s_n|X_{[n-m,n-1]}=s_{[n-m,n-1]})
  =\Proba(X_n=s_n|X_{n-1}=s_{n-1})\,.
\end{multline}
We denote by $A\Subset\NatNumZero$ the fact that $A$ is a \emph{finite subset} of $\NatNumZero$. The first statement is a factorisation of conditional probabilities over disjoint index blocks:
\begin{multline}\label{eq:markov:productFactorisation}
 \ForAll\emptyset\not=B_0,A_1,B_1,\dotsc,B_{m-1},A_m,B_m\Subset\NatNumZero
 ,\\A\cap B=\emptyset \text{ where }
 A:=\biguplus_{i=1}^m A_i \text{ and } B:=\biguplus_{i=0}^m B_i
 ,x_A\in\mathcal{S}_A
 ,x_B\in\mathcal{S}_B
 ,\\\Bigl(
  \ForAll i \in [m] :\,\,
  b_i^-:=\min(B_{i-1})<\min(A_i)
  ,\max(A_i)<\min(B_i)=:b_i^+
 \,\Bigr)
 :\\\Proba(X_A=x_A|X_B=x_B)
  =\prod_{i=1}^m
  \Proba(X_{A_i}=x_{A_i}|X_{b_i^-}=x_{b_i^-},X_{b_i^+}=x_{b_i^+})\,.
\end{multline}

Secondly, a \NameMarkov{} process retains the \NameMarkov{} property under a \emph{Cartesian conditioning}:
\begin{multline}\label{eq:markov:remainsMarkovUnderCartesianConditioning}
 \ForAll \emptyset\not=C\Subset\NatNumZero
 ,S_C:=\prod_{n\in C} S_n \text{ with } S_n\subseteq\mathcal{S}_n:\quad
 (X|X_C\in S_C)\text{ is \NameMarkov{}.}
\end{multline}

\begin{proof}
We need the intermediate statements
\begin{multline}\label{eq:markov:closestPastSingle}
 \ForAll n\in\NatNum
 ,\emptyset\not=B\subseteq\IntegerSet{0,n-1}
 ,x_n\in\mathcal{S}_n
 ,x_B\in\mathcal{S}_B
 :\\\Proba(X_n=x_n|X_B=x_B)
  =\Proba(X_n=x_n|X_{\max(B)}=x_{\max(B)})
\end{multline}
and
\begin{multline}\label{eq:markov:cartesianLast}
 \ForAll \emptyset\not=A,B\Subset\NatNumZero
 ,b:=\max(B)<\min(A)
 ,x_A\in\mathcal{S}_A
 ,S_B\subseteq\Set{x_b}\times\mathcal{S}_{B\setminus\Set{b}}
 :\\\Proba(X_A=x_A|X_B\in S_B)
  =\Proba(X_A=x_A|X_b=x_b)\,.
\end{multline}

Proof of~\eqref{eq:markov:closestPastSingle}: Let $b:=\max(B)$, $C:=\IntegerSet{b+1,n-1}$ and $D:=\IntegerSet{\min(B),b}\setminus B$. We use~\eqref{eq:markov:classic} to get
\begin{align*}
 &\FirstAlign
  \Proba(X_n=x_n|X_B=x_B)
 \\&=\frac%
  {\sum_{x_C,x_D} \Proba(X_n=x_n,X_C=x_C,X_B=x_B,X_D=x_D)}
  {\Proba(X_B=x_B)}
 \\&=\frac%
  {\displaystyle\sum_{x_C,x_D}
   \Proba(X_n=x_n,X_C=x_C|X_B=x_B,X_D=x_D)\Proba(X_B=x_B,X_D=x_D)
  }
  {\Proba(X_B=x_B)}
 \\&=\frac%
  {\displaystyle\sum_{x_C} \Proba(X_n=x_n,X_C=x_C|X_b=x_b)
  \sum_{x_D} \Proba(X_B=x_B,X_D=x_D)
  }
  {\Proba(X_B=x_B)}
 \\&= \Proba(X_n=x_n|X_b=x_b)\,.
\end{align*}

Proof of~\eqref{eq:markov:productFactorisation}: For $A\Subset\NatNumZero$, we abbreviate the event $E_A:=\Iverson{X_A=x_A}$. Apply~\eqref{eq:markov:closestPastSingle} to get

\begin{align*}
 &\FirstAlign
  \Proba(X_A=x_A|X_B=x_B)
 \\&=\frac{\Proba(X_A=x_A,X_B=x_B)}{\Proba(X_B=x_B)}
 \\&=\frac%
  {\displaystyle
   \Proba(E_{B_0})
   \prod_{i=1}^m
   \Proba(E_{B_i\setminus\Set{b_i^+}}|
        (E_{A_j})_{j\le i},(E_{B_j})_{j<i},E_{b_i^+})
   \Proba(E_{b_i^+},E_{A_i}|(E_{A_j})_{j<i},(E_{B_j})_{j<i})
  }
  {\displaystyle
   \Proba(E_{B_0})
   \prod_{i=1}^m
   \Proba(E_{B_i\setminus\Set{b_i^+}}|(E_{B_j})_{j<i},E_{b_i^+})
   \Proba(E_{b_i^+}|(E_{B_j})_{j<i})
  }
 \\&=\prod_{i=1}^m \frac%
  {
   \Proba(E_{B_i\setminus\Set{b_i^+}}|E_{b_i^+})
   \Proba(E_{b_i^+},E_{A_i}|E_{b_i^-})
  }
  {
   \Proba(E_{B_i\setminus\Set{b_i^+}}|E_{b_i^+})
   \Proba(E_{b_i^+}|E_{b_i^-})
  }
 \\&=\prod_{i=1}^m \Proba(E_{A_i}|E_{b_i^-},E_{b_i^+})\,.
\end{align*}

Proof of~\eqref{eq:markov:cartesianLast}: Let $b:=\max(B)$. We apply~\eqref{eq:markov:closestPastSingle} to get

\begin{align*}
 &\FirstAlign
  \Proba(X_A=x_A|X_B\in S_B)
 \\&= \frac%
  {\sum_{x_B\in S_B} \Proba(X_A=x_A,X_B=x_B)}
  {\Proba(X_B\in S_B)}
 \\&= \frac%
  {\sum_{x_B\in S_B}\Proba(X_A=x_A|X_{b}=x_{b})\Proba(X_B=x_B)}
  {\Proba(X_B\in S_B)}
 \\&=\Proba(X_A=x_A|X_{b}=x_{b})\,.
\end{align*}

Proof of~\eqref{eq:markov:remainsMarkovUnderCartesianConditioning}: Let $n,m\in\NatNum$ with $m\le n$. Let $B:=\IntegerSet{n-m,n-1}$, $x_n\in\mathcal{S}_n$ and $x_B\in\mathcal{S}_B$. Let $C_+:=C\setminus\IntegerSet{0,n-1}$ and $C_-:=C\cap\IntegerSet{0,n-1}$. Thus $S_C=S_{C_-}\times S_{C_+}$. We apply~\eqref{eq:markov:cartesianLast} twice to show that $(X|X_C\in S_C)$ fulfils~\eqref{eq:markov:classic} and is thus \NameMarkov{}:

\begin{align*}
 &\FirstAlign
  \Proba(X_n=x_n|X_B=x_B,X_C\in S_C)
 \\&= \frac%
  {\Proba(X_n=x_n,X_{C_+}\in S_{C_+},X_B=x_B,X_{C_-}\in S_{C_-})}
  {\Proba(X_B=x_B,X_C\in S_C)}
 \\&=\frac%
  {\Proba(X_n=x_n,X_{C_+}\in S_{C_+}|X_B=x_B,X_{C_-}\in S_{C_-})
  \Proba(X_B=x_B,X_{C_-}\in S_{C_-})
  }
  {\Proba(X_{C_+}\in S_{C_+}|X_B=x_B,X_{C_-}\in S_{C_-})
  \Proba(X_B=x_B,X_{C_-}\in S_{C_-})}
 \\&=\frac%
  {\Proba(X_n=x_n,X_{C_+}\in S_{C_+}|X_{n-1}=x_{n-1})}
  {\Proba(X_{C_+}\in S_{C_+}|X_{n-1}=x_{n-1})}
 \\&=\Proba(X_n=x_n|X_{n-1}=x_{n-1},X_{C_+}\in S_{C_+})
 \\&=\Proba(X_n=x_n|X_{n-1}=x_{n-1},X_{C_+}\in S_{C_+},X_{C_-}\in S_{C_-})
 \\&=\Proba(X_n=x_n|X_{n-1}=x_{n-1},X_C\in S_C)\,.
\end{align*}
\end{proof}
\subsection{The loss case}
\label{sec:proof:loss}
We start with some derivations common to the proof of~\eqref{eq:entropyrate:loss} and~\eqref{eq:trajectory:explosion}. We assume $\Kappa<\infty$. Equation~\eqref{eq:kappa} is equivalent to the existence of $\check{x},\hat{x}\in\SpaceX,\Sequence{y}\in\SpaceY^{\Kappa},\Sequence{x}\in \LumpingPreimage(\Sequence{y})$ with
\begin{equation}\label{eq:finite:setting}
 0
 <\Proba(X_0=\check{x},X_{\IntegerSet{\Kappa}}=\Sequence{x},X_{\Kappa+1}=\hat{x})
 <\Proba(X_0=\check{x},Y_{\IntegerSet{\Kappa}}=\Sequence{y},X_{\Kappa+1}=\hat{x})\,.
\end{equation}
Let $\Sequence{s}:=(\check{x},\Sequence{x},\hat{x})$. The \emph{unreconstructable set of trajectories} $\mathcal{H}$ is
\begin{equation}\label{eq:finite:lossSet}
 \mathcal{H}:=\Set{\check{x}}\times \LumpingPreimage(\Sequence{y})\times\Set{\hat{x}}\,.
\end{equation}
Equation~\eqref{eq:kappa} implies that $\mathcal{H}$ contains at least two elements with positive probability. If we pass through $\mathcal{H}$, then we incur an \emph{entropy loss} $L$:
\begin{equation}\label{eq:finite:minimumLoss}
 L:=\Entropy{X_{\IntegerSet{\Kappa}}|X_{\IntegerSet{0,\Kappa+1}}\in \mathcal{H}}>0\,.
\end{equation}
Let $\InstantsHiddenTraversal$ be the random set of indices marking the start of non-overlapping runs of $X_{\IntegerSet{n}}$ through $\mathcal{H}$, that is,
\begin{equation}\label{eq:finite:indices}
 \InstantsHiddenTraversal
 :=\BigSet{i\in \IntegerSet{n-\Kappa-1}:
  \begin{gathered}
   X_{\IntegerSet{i,i+\Kappa+1}}\in \mathcal{H}\\
   \text{and}\\
   \ForAll j\in \IntegerSet{i+1,i+\Kappa+1}:
   X_{\IntegerSet{j,j+\Kappa+1}}\not\in \mathcal{H}
  \end{gathered}
 }\,,
\end{equation}
where we select lower indices greedily. For the $\Sequence{s}$ from after~\eqref{eq:finite:setting}, we lower-bound the tail probability of the cardinality of $\InstantsHiddenTraversal$ by the one of $\InstantsNonOverLappingTraversal{X}{\Sequence{s}}{n}$:
\begin{equation}\label{eq:finite:indexSetComparison}
 \ForAll m\in\NatNum:\quad
 \Proba(\Cardinality{\InstantsHiddenTraversal}\ge m)
 \ge \Proba(\Cardinality{\InstantsNonOverLappingTraversal{X}{\Sequence{s}}{n}}\ge m)\,.
\end{equation}
Finally, let
\begin{equation}\label{eq:finite:alpha}
 \alpha:=\frac{\Proba(X_{\IntegerSet{\Kappa+2}}=\Sequence{s})}{2(\Kappa+2)}>0\,.
\end{equation}

\begin{proof}[Proof of~\eqref{eq:entropyrate:loss}]
We claim that, for every $m\in\NatNum$:
\begin{equation}\label{eq:finite:linearLoss}
 \Entropy{X_{\IntegerSet{n}}|Y_{\IntegerSet{n}}}
 \ge \Proba(\Cardinality{\InstantsHiddenTraversal}\ge m)
 \Entropy{X_{\IntegerSet{n}}|Y_{\IntegerSet{n}},\Cardinality{\InstantsHiddenTraversal}\ge m}
 \ge \Proba(\Cardinality{\InstantsHiddenTraversal}\ge m)\, m L\,.
\end{equation}
Combining~\eqref{eq:finite:linearLoss} and~\eqref{eq:finite:indexSetComparison}, for $m=\alpha n$, with ~\eqref{eq:lowerBound:probability}, we arrive at~\eqref{eq:entropyrate:loss}:
\begin{multline}\label{eq:loss:quantification}
 \FirstAlign\EntropyRate{X|Y}
 =
 \LimN\OneOverN\Entropy{X_{\IntegerSet{n}}|Y_{\IntegerSet{n}}}\\
 \ge
 \LimN\OneOverN \Proba(\Cardinality{\InstantsHiddenTraversal}\ge \alpha n)\, \alpha n L
 \ge
 \alpha L \LimN \Proba(\Cardinality{\InstantsNonOverLappingTraversal{X}{\Sequence{s}}{n}}\ge \alpha n)
 =
 \alpha L >0\,.
\end{multline}
It rests to prove~\eqref{eq:finite:linearLoss}. We fix $m,n\in\NatNum$. For $\RealisationHiddenTraversal\subseteq \IntegerSet{n}$ with $\Proba(\InstantsHiddenTraversal=\RealisationHiddenTraversal)>0$ and each $i\in\RealisationHiddenTraversal$, we derive the indices of the block $B_i:=\IntegerSet{i,i+\Kappa+1}$ and its interior $\widehat{B_i}:=\IntegerSet{i+1,i+\Kappa}$. Their unions are $B:=\biguplus_{i\in\RealisationHiddenTraversal} B_i$ and $\widehat{B}:=\biguplus_{i\in\RealisationHiddenTraversal} \widehat{B}_i$ respectively. Hence,
\begin{subequations}\label{eq:finite:steps}
\begin{align}
 \notag
 &\FirstAlign
  \Entropy{X_{\IntegerSet{n}}|Y_{\IntegerSet{n}},\InstantsHiddenTraversal=\RealisationHiddenTraversal}\,,
 \\\label{eq:finite:steps:outside}
 &\ge
  \Entropy{X_{\widehat{B}}|X_{\IntegerSet{n}\setminus B},\forall i\in\RealisationHiddenTraversal: X_{B_i}\in \mathcal{H}}\,,
 \\\label{eq:finite:steps:markov}
 &=
  \Entropy{X_{\widehat{B}}|\forall i\in\RealisationHiddenTraversal: X_{B_i}\in \mathcal{H}}\,,
 \\\label{eq:finite:steps:conditional}
 &=
  \sum_{i\in\RealisationHiddenTraversal}\Entropy{X_{\widehat{B_i}}|X_{B_i}\in \mathcal{H}}\,,
 \\\label{eq:finite:steps:minLoss}
 &= \Cardinality{\RealisationHiddenTraversal}\times L\,,
\end{align}
\end{subequations}
where in~\eqref{eq:finite:steps:outside} we throw away all information outside $\widehat{B}$ and condition on it, in~\eqref{eq:finite:steps:markov} we apply the conditional factorisation~\eqref{eq:markov:productFactorisation} to remove every condition except the block ends, in~\eqref{eq:finite:steps:conditional} we apply the conditional factorisation~\eqref{eq:markov:productFactorisation} to the \NameMarkov{} process $(X|X_B\in\mathcal{H}^{\Cardinality{\RealisationHiddenTraversal}})$ (as $\mathcal{H}$ is a cartesian product) and in~\eqref{eq:finite:steps:minLoss} we conclude by stationarity and the minimum loss~\eqref{eq:finite:minimumLoss}. Hence,
\begin{align*}
  \Entropy{X_{\IntegerSet{n}}|Y_{\IntegerSet{n}},\Cardinality{\InstantsHiddenTraversal}\ge m}
 &=
  \sum_{\substack{\RealisationHiddenTraversal\subseteq\IntegerSet{n}\\\Cardinality{\RealisationHiddenTraversal}\ge m}}
  \Proba(\InstantsHiddenTraversal=\RealisationHiddenTraversal|\Cardinality{\InstantsHiddenTraversal}\ge m)
  \Entropy{X_{\IntegerSet{n}}|Y_{\IntegerSet{n}},\InstantsHiddenTraversal=\RealisationHiddenTraversal}
 \\&\ge
  \sum_{\substack{\RealisationHiddenTraversal\subseteq\IntegerSet{n}\\\Cardinality{\RealisationHiddenTraversal}\ge m}}
  \Proba(\InstantsHiddenTraversal=\RealisationHiddenTraversal|\Cardinality{\InstantsHiddenTraversal}\ge m)
  \times\Cardinality{\RealisationHiddenTraversal}\times L
 \\&\ge mL\,.
\end{align*}

\end{proof}

\begin{proof}[Proof of~\eqref{eq:trajectory:explosion}]

For the $\Sequence{s}$ from after~\eqref{eq:finite:setting}, we have
\begin{equation}\label{eq:preimageCountVsPreimageTraversal}
 T_n \ge 2^{\Cardinality{\InstantsNonOverLappingTraversal{X}{\Sequence{s}}{n}}}\,.
\end{equation}
Thus,~\eqref{eq:preimageCountVsPreimageTraversal} and~\eqref{eq:lowerBound:almostSure} imply that
\begin{multline}\label{eq:trajectory:rate}
 \liminf_{n\to\infty} \sqrt[n]{T_n}
 \ge \liminf_{n\to\infty} \Exponential((\log 2)\OneOverN \Cardinality{\InstantsNonOverLappingTraversal{X}{\Sequence{s}}{n}})
 \\
 =\Exponential((\log 2)\liminf_{n\to\infty}\OneOverN \Cardinality{\InstantsNonOverLappingTraversal{X}{\Sequence{s}}{n}})
 \stackrel{\Proba-a.s.}{\ge} \Exponential((\log 2)\alpha)
 = 2^\alpha>1\,.
\end{multline}
\end{proof}

\begin{proof}[Proof of proposition~\ref{prop:kappaBound}]
Let $x_0,x_{\Kappa+1},\Sequence{y},\Sequence{x'},\Sequence{x''}$ be as in~\eqref{eq:kappa}. Suppose that $\Kappa> K:=\sum_{y\in\SpaceY}\Cardinality{\LumpingPreimage(y)}(\Cardinality{\LumpingPreimage(y)}-1)$ and $\Kappa>1$. We apply the \emph{pigeon-hole principle}, first to every $x\in\LumpingPreimage(y)$ and then to each $\LumpingPreimage(y)$, for every $y\in\Support{\Sequence{y}}$. This ensures that the two trajectories intersect:
\begin{equation}\label{eq:pigeonhole}
 \Exists m\in\IntegerSet{\Kappa}:
 \quad\Sequence{x'}_\Set{m}=\Sequence{x''}_\Set{m}\,.
\end{equation}
Choose $m$ fulfilling~\eqref{eq:pigeonhole}. If $m=1$, then $\Sequence{x'}_\Set{1},x_{\Kappa+1},\Sequence{y}_{\IntegerSet{2,\Kappa}},\Sequence{x'}_{\IntegerSet{2,\Kappa}},\Sequence{x''}_{\IntegerSet{2,\Kappa}}$ fulfil the conditions in~\eqref{eq:kappa}. If $m>1$, then $x_0,\Sequence{x'}_\Set{m},\Sequence{y}_{\IntegerSet{m-1}},\Sequence{x'}_{\IntegerSet{m-1}},\Sequence{x''}_{\IntegerSet{m-1}}$ fulfil the conditions in~\eqref{eq:kappa}. Both cases lead to $\Kappa<\Kappa$, a contradiction.
\end{proof}

\section{Proof of strong \texorpdfstring{$k$}{k}-lumpability}
\label{sec_proof_lumpability}
For (conditional) probabilities we use the following short-hand notation:
\begin{equation*}\label{eq:probaShortHand}
 \Proba(Z=z) = p_Z(z)
 \quad\text{ and }\quad
 \Proba(Z_1=z_1|Z_2=z_2) = p_{Z_1|Z_2}(z_1|z_2)\,,
\end{equation*}
where we always assume that the latter is well-defined, i.e., that $p_{Z_2}(z_2)>0$. Recall that the \emph{conditional mutual information} of $Z_1$ and $Z_2$ given $Z_3$ is
\begin{equation}\label{eq:condMutInf}
 \CondMutInf{Z_1}{Z_2}{Z_3}
 := \Entropy{Z_1|Z_3}-\Entropy{Z_1|Z_2,Z_3}\,.
\end{equation}
The conditional mutual information vanishes, iff $Z_1$ and $Z_2$ are conditionally independent given $Z_3$~\cite[Thm.~2.6.3]{Cover_Thomas__ElementsOfInformationTheory_Ed2__Wiley_2006}.

\begin{proof}[Proof of proposition~\ref{prop:ratetoMarkov}]
The rhs of~\eqref{eq:rateToMarkov} is equivalent to
\begin{align*}
 0&= \Entropy{Z_k|Z_{\IntegerSet{0,k-1}}}-\EntropyRate{Z}
 \\&= \Entropy{Z_k|Z_{\IntegerSet{0,k-1}}}-\LimN\Entropy{Z_n|Z_{\IntegerSet{0,n-1}}}
 \\&= \LimN \left(\Entropy{Z_n|Z_{\IntegerSet{n-k,n-1}}}-\Entropy{Z_n|Z_{\IntegerSet{0,n-1}}}\right)
 \\&= \LimN \CondMutInf{Z_n}{Z_{\IntegerSet{0,n-k-1}}}{Z_{\IntegerSet{n-k,n-1}}}\,.
\end{align*}
By stationarity, the sequence in the last limit increases monotonically in $n$. A limit value of zero is equivalent to, for all $n\in\NatNum$:
\begin{align*}
 &\FirstAlign p_{Z_n|Z_{\IntegerSet{n-k,n-1}}}(\cdot|\Sequence{z}) p_{Z_{\IntegerSet{0,n-k-1}}|Z_{\IntegerSet{n-k,n-1}}}(\cdot|\Sequence{z})
 \\&= p_{Z_n,Z_{\IntegerSet{0,n-k-1}}|Z_{\IntegerSet{n-k,n-1}}}(\cdot|\Sequence{z})
 \\&= p_{Z_n|Z_{\IntegerSet{0,n-1}}}(\cdot|\cdot,\Sequence{z}) p_{Z_{\IntegerSet{0,n-k-1}}|Z_{\IntegerSet{n-k,n-1}}}(\cdot|\Sequence{z})\,,
\end{align*}
where the first equality holds $p_{Z_{\IntegerSet{n-k,n-1}}}$-a.s. The equality between the first and last line is equivalent to the higher-order \NameMarkov{} property~\eqref{eq:markov:kthorder}.
\end{proof}

\begin{proof}[Proof of theorem~\ref{thm:klump:equalities}]
The equivalence in~\eqref{eq:klump} follows from the equivalence of its two statements to the following technical property:
\begin{multline}\label{eq:klump:probabilities}
 \ForAll y',y\in\SpaceY,\Sequence{y}\in\SpaceY^{k-1}, x\in\LumpingPreimage(y):\quad
 p_{Y_k,Y_{\IntegerSet{k-1}},X_0}(y',\Sequence{y},x)>0
 \Then
 \\0<p_{Y_k|Y_{\IntegerSet{k-1}},X_0}(y'|\Sequence{y},x)
 =p_{Y_k|Y_{\IntegerSet{k-1}},Y_0}(y'|\Sequence{y},y)\,.
\end{multline}

The equivalence between~\eqref{eq:klump:boundsequal} and~\eqref{eq:klump:probabilities} is in proposition~\ref{prop_klump_mutinf} and the equivalence between~\eqref{eq:klump:strong} and~\eqref{eq:klump:probabilities} is in proposition~\ref{prop_klump_strong}.\end{proof}

\begin{Prop}\label{prop_klump_mutinf}
 For a lumping $\Lumping$, property~\eqref{eq:klump:boundsequal} is equivalent to~\eqref{eq:klump:probabilities}.
\end{Prop}

\begin{proof}
We rewrite~\eqref{eq:klump:boundsequal} as
\begin{align*}
 0&=\Entropy{Y_k|Y_{\IntegerSet{0,k-1}}}-\Entropy{Y_k|Y_{\IntegerSet{k-1}},X_0}
 \\&=\Entropy{Y_k|Y_{\IntegerSet{0,k-1}}}-\Entropy{Y_k|Y_{\IntegerSet{0,k-1}},X_0}
 \\&=\CondMutInf{Y_k}{X_0}{Y_{\IntegerSet{0,k-1}}}\,.
\end{align*}
For all $y'\in\SpaceY,\Sequence{y}\in\SpaceY^k, x\in\SpaceX$ with $ p_{Y_k,Y_{\IntegerSet{0,k-1}},X_0}(y',\Sequence{y},x)>0$, this is equivalent to:
\begin{equation*}
 0< p_{Y_k,X_0|Y_{\IntegerSet{0,k-1}}}(\cdot|\Sequence{y})
 = p_{Y_k|Y_{\IntegerSet{0,k-1}}}(\cdot|\Sequence{y})p_{X_0|Y_{\IntegerSet{0,k-1}}}(\cdot|\Sequence{y})\,.
\end{equation*}
Division in the previous line equals~\eqref{eq:klump:probabilities}.
\end{proof}

\begin{Prop}\label{prop_klump_strong}
A lumping $\Lumping$ is strongly $k$-lumpable, iff~\eqref{eq:klump:probabilities} holds.
\end{Prop}

\begin{proof}
This is a straightforward generalization of the proof for the case $k=1$ in~\cite{Kemeny_Snell__FiniteMarkovChains__Springer_1976}. See section~\ref{sec_proof_lumpability}.
\end{proof}

\section{Proofs of the sufficient conditions}
\label{sec:proof:sufficient}
We use the shorthand notation introduced at the beginning of section~\ref{sec_proof_lumpability}.

\begin{proof}[Proof of proposition~\ref{prop:SESXpreserves}]
We have
\begin{equation*}
 \Entropy{Y_k|X_{k-1}}
 \leq\Entropy{Y_k|Y_{\IntegerSet{k-1}},X_0}
 \leq \EntropyRate{Y}
 \leq \EntropyRate{X}
 = \Entropy{X_k|X_{k-1}}\,,
\end{equation*}
where the first and the second inequality are due to~\cite[Thm.~4.5.1, pp.~86]{Cover_Thomas__ElementsOfInformationTheory_Ed2__Wiley_2006} (cf.~lemma~\ref{lem:ratebounds}) and the third inequality is due to data processing~\cite{Geiger_Kubin__SomeResultsOnTheInformationLossInDynamicalSystems__ISWSC_2011,Watanabe_Abraham__LossAndRecoveryOfInformationByCoarseObservationOfStochasticChain__IC_19}. The $\ClassSingleEntry$ property implies that $p_{X_k,X_{k-1}}$-a.s.
\begin{equation*}
 p_{Y_k|X_{k-1}}(y|x)
 =p_{X_k|X_{k-1}}(x'(x,y)|x)\,,
\end{equation*}
where $x'(x,y)$ is unique endpoint of the edge existing by~\eqref{eq:singleEntry}. Thus, the outer terms in the above chain of inequalities coincide, yielding $\EntropyRate{Y}=\EntropyRate{X}$.
\end{proof}

\begin{proof}[Proof of proposition~\ref{prop:UEimpliesklump}]
First, we show that $\ClassSFS{k}$ is a subclass of $\ClassSingleEntry$, implying preservation of entropy. If $\ClassSingleEntry$ does not hold, then there exist states $y^\star\in\SpaceY$ and $x^\star\in\SpaceX$ such that at least two states $x',x''\in\LumpingPreimage(y^\star)$ have positive transition probabilities from $x^\star$. Choose a realisable path $\Sequence{x}_{\IntegerSet{0,k-3}}$, with positive transition probability from $x_{k-3}$ to $x^\star$. Let $\Sequence{y}=(\LumpingFunction(\Sequence{x}_{\IntegerSet{k-3}}),\LumpingFunction(x^\star),y^\star)\in\SpaceY^{k-1}$. We have
\begin{equation*}
 p_{X_{\IntegerSet{k-1}}|Y_{\IntegerSet{k-1}},X_0}
 (\Sequence{x}_{\IntegerSet{k-3}},x^\star,x'|\Sequence{y},\Sequence{x}_{\Set{0}})>0
\end{equation*}
and
\begin{equation*}
 p_{X_{\IntegerSet{k-1}}|Y_{\IntegerSet{k-1}},X_0}
 (\Sequence{x}_{\IntegerSet{k-3}},x^\star,x''|\Sequence{y},\Sequence{x}_{\Set{0}})>0\,.
\end{equation*}
This contradicts the definition of $\ClassSFS{k}$~\eqref{eq:singleforwardsequence}.\\

Second, we show that $\ClassSFS{k}$ implies strong $k$-lumpability of $\Lumping$. We check~\eqref{eq:klump:probabilities} and then conclude via proposition~\ref{prop_klump_strong}.
We have $p_{Y_{\IntegerSet{k-1}},X_0}$-a.s. a unique $\Sequence{x'}(\LumpingFunction(X_0),\Sequence{Y}_{\IntegerSet{k-1}})\in\SpaceX^{k-1}$ fulfilling~\eqref{eq:singleforwardsequence}. Hence,
\begin{align*}
 &p_{Y_k|Y_{\IntegerSet{k-1}},X_0}
  (y|\Sequence{y},x)
 \\=&
 p_{Y_k|Y_{\IntegerSet{k-1}},X_{\IntegerSet{k-1}},X_0}
  (y|\Sequence{y},\Sequence{x'}(g(x),\Sequence{y}),x)
 \underbrace{
  p_{X_{\IntegerSet{k-1}}|Y_{\IntegerSet{k-1}},X_0}
  (\Sequence{x'}(g(x),\Sequence{y})|\Sequence{y},x)
 }_{=1\text{ by }\eqref{eq:singleforwardsequence}}
 \\=&p_{Y_k|X_{\IntegerSet{k-1}},X_0}(y|\Sequence{x'}(g(x),\Sequence{y}),x)
 \\=&p_{Y_k|X_{k-1}}(y|\Sequence{x'}(g(x),\Sequence{y})_{\Set{k-1}})
  \qquad\text{by the Markov property of }X
\end{align*}
is independent of $x$ and~\eqref{eq:klump:probabilities} holds.
\end{proof}

\section*{Acknowledgements}
\label{sec:acknowledgements}
We thank Gernot Kubin and Wolfgang Woess for establishing contact between us, leading to our joint investigation of this topic. We are particularly indebted to our anonymous reviewer for his thoughtful comments and encouragement to flesh out the part about $k$-lumpability.
\section{Additional Material}
\label{sec_additional}

\subsection{k-lumping details}
\begin{proof}[Proof of proposition~\ref{prop_klump_strong}]
From definitions~\ref{def:kthOrderMC} and~\ref{def:klump}, strong $k$-lumpability is:
\begin{multline}\label{eq_klumpproof_klump}
 \ForAll y,y'\in\SpaceY, \Sequence{y}\in\SpaceY^{k-1},\Sequence{s}\in\SpaceY^n,\nu \text{ distribution of $X_0$}:\\
  p^\nu_{Y_{n+k},Y_{\IntegerSet{n+1,n+k-1}},Y_n,Y_{\IntegerSet{0,n-1}}}
  (y',\Sequence{y},y,\Sequence{s})
  \\= \gamma(y'|\Sequence{y},y)
  p^\nu_{Y_{\IntegerSet{n+1,n+k-1}},Y_n,Y_{\IntegerSet{0,n-1}}}
  (\Sequence{y},y,\Sequence{s})\,,
\end{multline}
where $p^\nu$ is the distribution of the indicated rvs given $\nu$ as distribution of $X_0$ and where $\gamma$ is the $k$-th order \NameMarkov{} transition kernel.\\

Strong $k$-lumpability implies~\eqref{eq:klump:probabilities}: Strong $k$-lumpability implies that
\begin{equation*}
 p^\nu_{Y_k|Y_{\IntegerSet{k-1}},Y_0}(y'|\Sequence{y},y)=\gamma(y'|\Sequence{y},y)\,,
\end{equation*}
independently of the distribution $\nu$ of $X_0$. We follow~\cite[Thm.~6.3.2]{Kemeny_Snell__FiniteMarkovChains__Springer_1976} and let $x\in\LumpingPreimage(y)$ such that $p_{Y_k,Y_{\IntegerSet{k-1}},X_0}(y',\Sequence{y},x)>0$. Then the following summation degenerates to a single summand:
\begin{align*}
 &\FirstAlign \gamma(y'|\Sequence{y},y)
 \\&=p^\nu_{Y_k|Y_{\IntegerSet{k-1}},Y_0}(y'|\Sequence{y},y)
 &\text{by~\eqref{eq_klumpproof_klump}}
 \\&=\sum_{x'}
  p_{Y_k|Y_{\IntegerSet{k-1}},Y_0,X_0}(y'|\Sequence{y},y,x')
  p^\nu_{X_0|Y_{\IntegerSet{k-1}},Y_0}(x'|\Sequence{y},y)
 \\&= p_{Y_k|Y_{\IntegerSet{k-1}},Y_0,X_0}(y'|\Sequence{y},y,x)
  p^{\delta_x}_{X_0|Y_{\IntegerSet{k-1}},Y_0}(x|\Sequence{y},y)
 &\text{by choosing $\nu=\delta_x$}
 \\&=p_{Y_k|Y_{\IntegerSet{k-1}},Y_0,X_0}(y'|\Sequence{y},y,x)
 \\&>0\,.
\end{align*}

\eqref{eq:klump:probabilities} implies strong $k$-lumpability: Choose $\Sequence{s}\in\SpaceY^n, y,y'\in\SpaceY, \Sequence{y}\in\SpaceY^{k-1}$ and a starting distribution $\nu$. We use the abbreviations $\alpha_x$ and $\beta_x$ from~\eqref{eq_klumpproof_pathspresent}. Using Chapman-Kolmogorov, we split
\begin{align}
  &p^\nu_{Y_{n+k},Y_{\IntegerSet{n+1,n+k-1}},Y_n,Y_{\IntegerSet{0,n-1}}}
  (y',\Sequence{y},y,\Sequence{s})
 \notag\\=
 & \sum_{x\in\LumpingPreimage(y)}
  p^\nu_{Y_{n+k},Y_{\IntegerSet{n+1,n+k-1}},X_n,Y_{\IntegerSet{0,n-1}}}
  (y',\Sequence{y},x,\Sequence{s})
 \notag\\=
 & \sum_{x\in\LumpingPreimage(y)}
  p^\nu_{Y_{n+k},Y_{\IntegerSet{n+1,n+k-1}}|X_n,Y_{\IntegerSet{0,n-1}}}
  (y',\Sequence{y}|x,\Sequence{s})
  p^\nu_{X_n,Y_{\IntegerSet{0,n-1}}}(x,\Sequence{s})
 \notag\\=
 & \sum_{x\in\LumpingPreimage(y)}
  p_{Y_{n+k},Y_{\IntegerSet{n+1,n+k-1}}|X_n}(y',\Sequence{y}|x)
  p^\nu_{X_n,Y_{\IntegerSet{0,n-1}}}(x,\Sequence{s})
 \label{eq_klumpproof_startdisloss}
 \\=
 & \sum_{x\in\LumpingPreimage(y)}
  \beta_x \Iverson{\beta_x>0} \alpha_x
  p^\nu_{X_n,Y_{\IntegerSet{0,n-1}}}(x,\Sequence{s})\,.\notag
\end{align}
The key step is~\eqref{eq_klumpproof_startdisloss}, where the conditioning on $x$ and the \NameMarkov{} property of $X$ let us discard the dependence on the starting distribution $\nu$ and time $n$. At this point we invoke~\eqref{eq:klump:probabilities} to see that $\beta_x$ is independent of $x$ and equals $p_{Y_k|Y_{\IntegerSet{k-1}},Y_0}(y'|\Sequence{y},y)$. Furthermore,~\eqref{eq_klumpproof_pathspresent} ensures that we can discard the indicator $\Iverson{\beta_x>0}$ in the remaining sum. Hence, we sum via Chapman-Kolmogorov and get
\begin{align*}
 &p^\nu_{Y_{n+k},Y_{\IntegerSet{n+1,n+k-1}},Y_n,Y_{\IntegerSet{0,n-1}}}
  (y',\Sequence{y},y,\Sequence{s})
 \\=& p_{Y_k|Y_{\IntegerSet{k-1}},Y_0}(y'|\Sequence{y},y)
  \sum_{x\in\LumpingPreimage(y)}
  \alpha_x p^\nu_{X_n,Y_{\IntegerSet{0,n-1}}}(x,\Sequence{s})
 \\=& p_{Y_k|Y_{\IntegerSet{k-1}},Y_0}(y'|\Sequence{y},y)
  p^\nu_{Y_{\IntegerSet{n+1,n+k-1}},Y_n,Y_{\IntegerSet{0,n-1}}}
  (\Sequence{y},y,\Sequence{s})\,.
\end{align*}
This equality holds for all $\nu$ and $n$, hence $Y$ is a $k$-th order \NameMarkov{} chain with transition kernel
\begin{equation*}
 \gamma(y'|\Sequence{y},y):=
 \begin{cases}
  p_{Y_k|Y_{\IntegerSet{k-1}},Y_0}(y'|\Sequence{y},y)
  & \text{if } p_{Y_k,Y_{\IntegerSet{k-1}},Y_0}(y',\Sequence{y},y)>0\,,
  \\
  0&\text{else.}
 \end{cases}
\end{equation*}
\end{proof}

\begin{Prop}\label{prop_klump_technical}
If $y,y'\in\SpaceY,\Sequence{y}\in\SpaceY^{k-1}$ with $p_{Y_k,Y_{\IntegerSet{k-1}},Y_0}(y',\Sequence{y},y)>0$ and~\eqref{eq:klump:probabilities} holds, then $\ForAll x\in\LumpingPreimage(y)$:
\begin{equation}\label{eq_klumpproof_pathspresent}
 \alpha_x:=p_{Y_{\IntegerSet{k-1}}|X_0}(\Sequence{y}|x)>0
 \Then
 \beta_x:=p_{Y_k|Y_{\IntegerSet{k-1}},X_0}(y'|\Sequence{y},x)>0\,.
\end{equation}
\end{Prop}

\begin{proof}
This follows from
\begin{align*}
 0 &<p_{Y_k|Y_{\IntegerSet{k-1}},Y_0}(y'|\Sequence{y},y)
 \\&=\frac%
  {p_{Y_k,Y_{\IntegerSet{k-1}},Y_0}(y',\Sequence{y},y)}
  {p_{Y_{\IntegerSet{k-1}},Y_0}(\Sequence{y},y)}
 \\&=\frac{
   \sum_{x\in\LumpingPreimage(y)}
   \beta_x \Iverson{\beta_x>0}
   \alpha_x \MeasureX(x)
  }{
   \sum_{x\in\LumpingPreimage(y)}
   \alpha_x \MeasureX(x)
  }
 \\&=p_{Y_k|Y_{\IntegerSet{k-1}},Y_0}(y'|\Sequence{y},y)
  \frac{
   \sum_{x\in\LumpingPreimage(y)}
   \Iverson{\beta_x>0} \alpha_x \MeasureX(x)
  }{
   \sum_{x\in\LumpingPreimage(y)}
   \alpha_x \MeasureX(x)
  }
 \,,
\end{align*}
where we apply~\eqref{eq:klump:probabilities} in the last equality to see that $\beta_x$ is constant on $\LumpingPreimage(y)$ and factor it out. Dividing both sides we get
\begin{equation*}
 1 = \frac{
   \sum_{x\in\LumpingPreimage(y)}
   \Iverson{\beta_x>0} \alpha_x \MeasureX(x)
  }{
   \sum_{x\in\LumpingPreimage(y)}
   \alpha_x \MeasureX(x)
  }\,.
\end{equation*}
As $\MeasureX$, the invariant measure, is positive, it follows that $\alpha_x>0$ implies $\beta_x>0$.\\
\end{proof}
\subsection{g-observability and entropy rate preservation}
\label{sec_gObservability}
This section gives a series of examples showing that $g$-observability, as defined in~\cite[Section~3]{Gurvits_Ledoux__MarkovPropertyForAFunctionOfAMarkovChain_ALinearAlgebraApproach}, is independent of entropy rate preservation.

\begin{Exam}[$\Kappa<\infty$ and $g$-observable]\label{exam_rateLoss_gObservable}
Regard the lumping represented by the transition matrix $P$ below, with bars marking the lumping. Let $0<\varepsilon<1/2$.
\begin{equation*}
 P := \left[\begin{array}{cc|c|c}
 0   & 0   & 1-\varepsilon & \varepsilon
 \\
 0   & 0   & \varepsilon   & 1-\varepsilon
 \\\hline
 1/3 & 1/3 & 1/3           & 0
 \\\hline
 1/3 & 1/3 & 0             & 1/3
 \end{array}\right]\,.
\end{equation*}
We have $\Kappa=1$, because of the paths $3-1-3$ and $3-2-3$. Using notation from~\cite[Section~3]{Gurvits_Ledoux__MarkovPropertyForAFunctionOfAMarkovChain_ALinearAlgebraApproach}, we let $\nu_\alpha:=\alpha\mathbb{I}_{1}+(1-\alpha)\mathbb{I}_{2}$. We have
\begin{multline*}
 V_g Q_n \nu_\alpha
 = \alpha V_g Q_n \mathbb{I}_{1} + (1-\alpha)V_g Q_n \mathbb{I}_{2}
 \\
 = \alpha (1-\varepsilon) 2^{-n} + (1-\alpha) \varepsilon 2^{-n}
 = 2^{-n} \left(\alpha(1-2\varepsilon)+\varepsilon \right)\,.
\end{multline*}
The last expression is clearly injective in $\alpha$, whence we are $g$-observable.
\end{Exam}

\begin{Exam}[$\Kappa<\infty$ and $g$-nonobservable]\label{exam_rateLoss_nonGObservable}
Regard the lumping represented by the transition matrix $P$ below, with bars marking the lumping:
\begin{equation*}
 P := \left[\begin{array}{cc|c}
 0   & 0   & 1/2
 \\
 0   & 0   & 1/2
 \\\hline
 1/3 & 1/3 & 1/3
 \end{array}\right]\,.
\end{equation*}
We have $\Kappa=1$, because of the paths $3-1-3$ and $3-2-3$. Not $g$-observable, because you can never recover the starting distribution in $\Set{1,2}$ from $Y$'s trajectory, as all trajectories starting with $Y_0=1$ (i.e. $X_0\in\Set{1,2}$) have $Y_1=2$ (i.e. $X_1=3$).
\end{Exam}

\begin{Exam}[$\Kappa=\infty$ and $g$-nonobservable]\label{exam_ratePreservation_nonGObservable}
Look at figure~\ref{fig:SESXCounterExample} (page~\ref{fig:SESXCounterExample}). It has $\Kappa=\infty$. But if we start in $Y_0=D$, we can not reconstruct if we are in $d_1$ or $d_2$ (as in example~\ref{exam_rateLoss_nonGObservable}).
\end{Exam}

\begin{Exam}[$\Kappa=\infty$ and $g$-observable]\label{exam_ratePreservation_gObservable}
Take an irreducible an aperiodic \NameMarkov{} chain with at least two states. Let $g$ be the identity mapping. Then $\Kappa=\infty$ and we can always reconstruct the starting distribution.
\end{Exam}
\subsection{Algorithmic aspects}
\label{sec_algorithmic}
This section explains in more detail the algorithmic upper bounds discussed in section~\ref{sec:discussion}.\\

Calculation of $\Kappa$: We have an upper bound of $\Kappa\le \Cardinality{\SpaceX}^2$ from~\eqref{eq:kappaBound}. Hence, there are at most $\Cardinality{\SpaceY}^{\Cardinality{\SpaceX}^2}$ paths of length $\Kappa$ for $Y$. As long as $n\le\Kappa+1$, there are at most $\Cardinality{\SpaceY}$ paths in the realisable preimage of each length-$n$ path in $\SpaceY$. This yields the $1$ in the exponent.\\

Calculation of other quantities: all other quantities need to evaluate some path probabilities under the invariant measure (needing at most $\Cardinality{\SpaceX}^2$ steps). There are at most $\Cardinality{\SpaceX}^n$ paths of length $n$. For $\ClassSingleEntry$, $\ClassSFS{k}$ and~\eqref{eq:klump:boundsequal} we need paths of lengths $2$, $k$ and $(k+1)$ respectively.\\

Better bounds should be attainable for the combinatorial conditions $\Kappa$, $\ClassSingleEntry$ and $\ClassSFS{k}$. This is due to the fact, that we are looking for violations of conditions imposing certain sparsity constraints on the $\TransitionGraph$. Thus, either the check finishes faster or fails with a violation of a constraint. A first flavour of this is in the above comment on the algorithmic bound for $\Kappa$.

\subsection{Reversed Processes}
\label{sec_reversed}
Equivalent conditions can be given for the transition matrix $\hat{P}$ of the reverse Markov chain $\hat{X}$. In other words, if either $\Lumping$ or $(\hat{P},\LumpingFunction)$ fulfil the conditions, preservation of entropy can be guaranteed.

\begin{Prop}[{\cite[Problem~4.2]{Cover_Thomas__ElementsOfInformationTheory_Ed2__Wiley_2006}}]
\label{prop:reverseRate}
The entropy rate of a stationary process $Z$ and its reverse process $\hat{Z}:=(\hat{Z}_n)_{n\in\IntNum}$, with $\hat{Z}_n:=Z_{-n}$ are the same.
\end{Prop}

\begin{proof}
\begin{align*}
 &\FirstAlign \EntropyRate{Z}
\\&= \LimN \Entropy{Z_n|Z_{[n[}}
\\&= \LimN \OneOverN \Entropy{Z_{\IntegerSet{n}}}
\\&= \LimN \OneOverN \sum_{i=0}^{n-1} \Entropy{Z_{n-i}|Z_{[n-i+1,n]}}
\\&= \LimN \Entropy{Z_1|Z_{[2,n]}}
\\&= \LimN \Entropy{\hat{Z}_n|\hat{Z}_{[n[}}
\\&= \EntropyRate{\hat{Z}}\,.
\end{align*}
\end{proof}

\begin{Def}\label{def:timeReversalOfStationaryMC}
The time reversal of stationary \NameMarkov{} chain is  a stationary \NameMarkov{} chain. If $P$ is the transition matrix of the original \NameMarkov{} chain $X$, then the transition matrix $\hat{P}$ of the reverse \NameMarkov{} chain $\hat{X}$ fulfils~\cite[Def.~5.3.1]{Kemeny_Snell__FiniteMarkovChains__Springer_1976}
\begin{equation}\label{eq:timeReversalOfStationaryMC}
 \hat{P}_{i,j}=\frac{\mu_jP_{j,i}}{\mu_i}\,.
\end{equation}
\end{Def}

\begin{Cor}[to proposition~\ref{prop:reverseRate}]\label{cor:reverseRate}
If a \NameMarkov{} chain $X$ can be lumped without information (rate) loss, then so can the reverse \NameMarkov{} chain $\hat{X}$.
\end{Cor}

\begin{proof}
Since the entropy rate does not change under reversing the process, and since a function $\LumpingFunction$ of the reverse \NameMarkov{} chain $\hat{X}$ is the reverse of the process $Y$, the result follows.
\end{proof}

\begin{Prop}\label{prop:kMarkovReverse}
 If a stationary process $Y$ is $k$-th order Markov, then so is the reverse process $\hat{Y}$. Thus, if a stationary Markov chain $X$ is $k$-lumpable, then so is the reverse chain $\hat{X}$.
\end{Prop}

\begin{proof}
 We start with showing that $Y\text{ is }\ClassKMarkov \Iff \hat{Y}\sim\ClassKMarkov$. To this end,
\begin{align*}
 &\FirstAlign \EntropyRate{Y}
\\&= \Entropy{Y_k|Y_{[0,k[}}
\\&= \Entropy{Y_{[0,k]}}-\Entropy{Y_{[0,k[}}
\\&= \Entropy{Y_{[0,k]}}-\Entropy{Y_{[k]}}
\\&= \Entropy{Y_0|Y_{[k]}}
\\&\ge \LimN \Entropy{Y_0|Y_{\IntegerSet{n}}}
\\&= \EntropyRate{\hat{Y}}
\end{align*}
which equals $\EntropyRate{Y}$ by Proposition~\ref{prop:reverseRate}. But for $Y\text{ is }\ClassKMarkov$ we need $\Lumping$ to be strongly $k$-lumpable by Definition~\ref{def:klump}. Since $\hat{Y}$ is obtained by lumping $(\hat{P},g)$ the proof follows.
\end{proof}

\begin{Exam}[taken from {\cite[pp.~139]{Kemeny_Snell__FiniteMarkovChains__Springer_1976}}]\label{exam:boundsnotequal}
Consider the following transition matrix, where the lines divide lumped states:
\begin{equation*}
 P := \left[\begin{array}{c|ccc}
 1/4  & 1/16  & 3/16 & 1/2\\
 \hline
 0    & 1/12  & 1/12 & 5/6\\
 0    & 1/12  & 1/12 & 5/6\\
 7/8  & 1/32  & 3/32 & 0
 \end{array}\right]\,.
\end{equation*}
This lumping (and its time-reversal) is weakly $1$-lumpable~\cite[pp.~139]{Kemeny_Snell__FiniteMarkovChains__Springer_1976}. However, we have (with an accuracy of $0.0001$)
\begin{equation*}
 0.5588 = \Entropy{Y_1|X_0}<\EntropyRate{Y}=\Entropy{Y_1|Y_0}= 0.9061
\end{equation*}
and
\begin{equation*}
 0.9048 = \Entropy{Y_0|X_1}<\EntropyRate{\hat{Y}}=\Entropy{Y_0|Y_1}= 0.9061\,,
\end{equation*}
where $\hat{Y}$ is the time-reversed process. Hence, weak $k$-lumpability alone does not imply~\eqref{eq:klump}.
\end{Exam}

\subsection{Blackwell's entropy rate expression}
\label{sec_blackwell}

Translation of the abstract of~\cite{Blackwell__TheEntropyOfFunctionsOfFiniteStateMarkovChains__CAS_1957} into present notation:\\

Let $\Set{X_n}_{n\in\IntNum}$ be a stationary ergodic finite-state \NameMarkov{} process with state space $\SpaceX$ and transition matrix $(m(x\to x'))_{x,x'\in\SpaceX}$. Let $g$ be a function defined on $\SpaceX$ with values in $\SpaceY$, and let $Y_n:=g(X_n)$. Then $\Set{Y_n}_{n\in\IntNum}$ is an ergodic stationary process, the general formula for the entropy of such processes being $H=-\Expect(\ld\Proba(Y_1|Y_0,Y_{-1},\dotsc))$. Let $\Set{A_{n,x}}_{n\in\IntNum,x\in\SpaceX}$, where $A_{n,x}:=\Proba(X_n=x|Y_n,Y_{n-1},\dotsc)$. It is shown that $\Set{A_{n,x}}_{n\in\IntNum,x\in\SpaceX}$ is a stationary \NameMarkov{} process with stationary distribution $Q$, where $Q$ is a distribution on vectors $(w_x)_{x\in\SpaceX}$, $\sum_{x\in\SpaceX} w_x=1$, $w_x\ge 0$, satisfying
\begin{equation}\label{eq_blackwell_measure}
 Q(E)=\sum_{y} \int_{\Set{f_y(w)\in E}} r_y(w) dQ(w)\,,
\end{equation}
 where $r_y(w):=\sum_{x}\sum_{x'\in\LumpingPreimage(y)} w_x m(x\to x')$, and $f_y$ is a vector function of $w$ whose $x$-th component is equal to $0$ if $\LumpingFunction(x)\not=y$, and equal to $\sum_{x'\in\LumpingPreimage(y)} w_x m(x\to x') / r_y(w)$ if $\Lumping(x)=y$. Then the entropy of $\Set{Y_n}_{n\in\IntNum}$ is given by
\begin{equation*}\label{eq_blackwell_entropy}
 H=-\int_{w} \sum_{y\in\SpaceY} r_y(w)\ld r_y(w) dQ(w)\,.
\end{equation*}
Under additional conditions it is shown that $Q$ is the only probability distribution that is a solution of~\eqref{eq_blackwell_measure} and that if $Q$ is continuous it is in a certain sense singular.\\

The expression~\eqref{eq_blackwell_entropy} involves an invariant measure on the simplex over $\SpaceX$. Furthermore, it is the invariant measure of a \NameMarkov{} chain involving the limit expressions $A_{n,x}$. This seems impossible to calculate in practice. If the processes live on time $\NatNumZero$ instead of $\IntNum$, then an equivalent of $A$ is difficult to define; $\Proba(X_n=x|Y_n,\dotsc,Y_0)$ as an expected value might not even be time-homogeneous any more.

\bibliographystyle{plain}
\bibliography{ref}

\end{document}